\documentclass{article}
\pdfoutput=1
\usepackage{amsmath}
\usepackage{amsfonts}
\usepackage{amsthm}
\usepackage{amssymb}
\usepackage{algorithmic}
\usepackage{psfrag}
\usepackage{fullpage}
\usepackage{graphicx}

\theoremstyle{definition}
\newtheorem{theorem}{Theorem}[section]
\newtheorem{lemma}[theorem]{Lemma}
\newtheorem{corollary}[theorem]{Corollary}

\newtheorem{definition}{Definition}

\newtheorem{observation}[theorem]{Observation}

\newtheorem{construction}[theorem]{Construction}

\newtheorem{notation}{Notation}
\newcommand{\Real}{\mathbb R}

\newcommand{\To}{\longrightarrow}

\newcommand{\Z}{\mathbb{Z}}
\newcommand{\N}{\mathbb{N}}
\newcommand{\func}[3]{#1 : #2 \rightarrow #3 }
\newcommand{\pfunc}[3]{#1 : #2 \dashrightarrow #3 }
\newcommand{\pair}[2]{(#1, #2)}

\newcommand{\color}{{\rm col}}

\newcommand{\strength}{{\rm str}}

\newcommand{\bval}[1]{[\![ #1 ]\!]}
\newcommand{\bbval}[1]{\left[\!\left[ #1 \right] \! \right]}
\newcommand{\dom}{{\rm dom} \;}
\newcommand{\asmb}{\mathcal{A}}
\newcommand{\asmbt}[2]{\asmb^{#1}_{#2}}
\newcommand{\asmbtt}{\mathcal{A}^\tau_T}
\newcommand{\ste}[2]{#1 \mapsto #2}
\newcommand{\frontier}[3]{{\partial}^{#1}_{#2}{#3}}
\newcommand{\frontiert}[1]{\partial^{\tau}{#1}}
\newcommand{\frontiertt}[1]{\frontier{\tau}{t}{#1}}

\newcommand{\frontiertau}[1]{{\partial}^{\tau}{#1}}
\newcommand{\arrowstett}[2]{#1 \xrightarrow[\tau,T]{1} #2}
\newcommand{\arrowste}[2]{#1 \stackrel{1}{\To} #2}
\newcommand{\arrowtett}[2]{#1 \xrightarrow[\tau,T]{} #2}
\newcommand{\res}[1]{\textrm{res}(#1)}
\newcommand{\termasm}[1]{\mathcal{A}_{\Box}\left[\mathcal{#1}\right]}
\newcommand{\prodasm}[1]{\mathcal{A}\left[\mathcal{#1}\right]}
\newcommand{\fgg}[1]{G^\#_{#1}}
\newcommand{\ftdepth}[2]{\textrm{ft-depth}_{#1}\left(#2\right)}

\begin{document}

%\begin{frontmatter}

\title{Strict Self-Assembly of Discrete Sierpinski Triangles}%
\author{James I. Lathrop\footnote{ Department of Computer
Science, Iowa State University, Ames, IA 50011, USA. jil@cs.iastate.edu. },
Jack H. Lutz\footnote{ Department of Computer Science, Iowa State University,
Ames, IA 50011, USA. lutz@cs.iastate.edu. This author's research was supported
in part by National Science Foundation Grants 0344187, 0652569, and 0728806 and
in part by Spanish Government MEC Project TIN 2005-08832-C03-02. }, and Scott
M. Summers.\footnote{Department of Computer Science, Iowa State University,
Ames, IA 50011, USA. summers@cs.iastate.edu. This author's research was
supported in part by NSF-IGERT Training Project in Computational Molecular
Biology Grant number DGE-0504304}}
%\date{}
\maketitle

%\clearpage

%\title{Strict Self-Assembly of Discrete Sierpinski Triangles}
%
%\author[ISU]{James I. Lathrop}
%\ead{jil@cs.iastate.edu}
%\author[ISU]{Jack H. Lutz\corauthref{JHL}\thanksref{label2}}
%\corauth[JHL]{Corresponding author.} \ead{lutz@cs.iastate.edu}
%\thanks[label2]{This author's research was supported in part by
%National Science Foundation Grants 0344187, 0652569, and 0728806 and
%in part by Spanish Government MEC Project TIN 2005-08832-C03-02.}
%\author[ISU]{Scott M. Summers\thanksref{label3}}
%\ead{summers@cs.iastate.edu}
%\thanks[label3]{This
%author's research was supported in part by NSF-IGERT Training
%Project in Computational Molecular Biology Grant number
%DGE-0504304.}
%
%\address[ISU]{Department of Computer
%Science, Iowa State University, Ames, IA 50011, USA.}

%\date{}

%\clearpage

%\maketitle
% ----------------------------------------------------------------
\begin{abstract}
Winfree (1998) showed that discrete Sierpinski triangles can
self-assemble in the Tile Assembly Model.  A striking molecular
realization of this self-assembly, using DNA tiles a few nanometers
long and verifying the results by atomic-force microscopy, was
achieved by Rothemund, Papadakis, and Winfree (2004).

Precisely speaking, the above self-assemblies tile completely filled-in,
two-dimensional regions of the plane, with labeled subsets of these tiles
representing discrete Sierpinski triangles. This paper addresses the more
challenging problem of the {\it strict self-assembly} of discrete Sierpinski
triangles, i.e., the task of tiling a discrete Sierpinski triangle and nothing
else.

We first prove that the standard discrete Sierpinski triangle {\it
cannot} strictly self-assemble in the Tile Assembly Model.  We then
define the {\it fibered Sierpinski triangle}, a discrete Sierpinski
triangle with the same fractal dimension as the standard one but
with thin fibers that can carry data, and show that the fibered
Sierpinski triangle strictly  self-assembles in the Tile Assembly
Model.  In contrast with the simple XOR algorithm of the earlier,
non-strict self-assemblies, our strict self-assembly algorithm makes
extensive, recursive use of optimal counters, coupled with measured
delay and corner-turning operations.  We verify our strict
self-assembly using the local determinism method of Soloveichik and
Winfree (2007).
\end{abstract}

%\end{frontmatter}

\section{Introduction}
Structures that self-assemble in naturally occurring biological
systems are often fractals of low dimension, by which we mean that
they are usefully modeled as fractals and that their fractal
dimensions are less than the dimension of the space or surface that
they occupy.  The advantages of such fractal geometries for
materials transport, heat exchange, information processing, and
robustness imply that structures engineered by nanoscale
self-assembly in the near future will also often be fractals of low
dimension.

The simplest mathematical model of nanoscale self-assembly is the
Tile Assembly Model (TAM), an extension of Wang tiling
\cite{Wang61,Wang63} that was introduced by Winfree \cite{Winf98}
and refined by Rothemund and Winfree \cite{RotWin00,Roth01}.  (See
also \cite{Adle99,Reif02,SolWin07}.) This elegant model, which is
described in section 2, uses tiles with various types and strengths
of ``glue'' on their edges as abstractions of molecules adsorbing to
a growing structure.  (The tiles are squares in the two-dimensional
TAM, which is most widely used, cubes in the three-dimensional TAM,
etc.) Despite the model's deliberate oversimplification of molecular
geometry and binding, Winfree \cite{Winf98} proved that the TAM is
computationally universal in two or more dimensions.  Self-assembly
in the TAM can thus be directed algorithmically.

This paper concerns the self-assembly of fractal structures in the
Tile Assembly Model.  The typical test bed for a new research topic
involving fractals is the Sierpinski triangle, and this is certainly
the case for fractal self-assembly.  Specifically,
Winfree~\cite{Winf98} showed that the {\it standard discrete
Sierpinski triangle} ${\mathbf S}$, which is illustrated in
Figure~1, self-assembles from a set of seven tile types in the Tile
Assembly Model.  Formally, ${\mathbf S}$ is a set of points in the
discrete Euclidean plane $\Z^2$.  The obvious and well-known
resemblance between ${\mathbf S}$ and the Sierpinski triangle in
$\Real^2$ that is studied in fractal geometry \cite{Falc03} is a
special case of a general correspondence between ``discrete
fractals'' and ``continuous fractals'' \cite{Will84}. Continuous
fractals are typically bounded (in fact, compact) and have intricate
structure at arbitrarily small scales, while discrete fractals like
$\mathbf{S}$ are unbounded and have intricate structure at
arbitrarily large scales.

A striking molecular realization of Winfree's self-assembly of
$\mathbf{S}$ was reported in 2004. Using DNA double-crossover
molecules (which were first synthesized in pioneering work of Seeman
and his co-workers \cite{Seem82}) to construct tiles only a few
nanometers long, Rothemund, Papadakis and Winfree \cite{RoPaWi04}
implemented the molecular self-assembly of $\mathbf{S}$ with low
enough error rates to achieve correct placement of 100 to 200 tiles,
confirmed by atomic force microscopy (AFM). This gives strong
evidence that self-assembly can be algorithmically directed at the
nanoscale.

The abstract and laboratory self-assemblies of $\mathbf{S}$
described above are impressive, but they are not (nor were they
intended or claimed to be) true fractal self-assemblies. Winfree's
abstract self-assembly of $\mathbf{S}$ actually tiles an {\it entire
quadrant} of the plane in such a way that five of the seven tile
types occupy positions corresponding to points in $\mathbf{S}$.
Similarly, the laboratory self-assemblies tile completely filled-in,
two-dimensional regions, with DNA tiles at positions corresponding
to points of $\mathbf{S}$ marked by inserting hairpin sequences for
AFM contrast. To put the matter figuratively, what self-assembles in
these assemblies is not the fractal $\mathbf{S}$ but rather a
two-dimensional canvas on which $\mathbf{S}$ has been painted.

In order to achieve the advantages of fractal geometries mentioned
in the first paragraph of this paper, we need self-assemblies that
construct fractal shapes {\it and nothing more}. Accordingly, we say
that a set $F\subseteq \mathbb{Z}^2$ {\it strictly self-assembles}
in the Tile Assembly Model if there is a (finite) tile system that
eventually places a tile on each point of $F$ and never places a
tile on any point of the complement, $\mathbb{Z}^2-F$. (This
condition is defined precisely in section 2.)

The specific topic of this paper is the strict self-assembly of
discrete Sierpinski triangles in the Tile Assembly Model. We present
two main results on this topic, one negative and one positive.

Our negative result is that the standard discrete Sierpinski
triangle $\mathbf{S}$ {\it cannot} strictly self-assemble in the
Tile Assembly Model. That is, there is no tile assembly system that
places tiles on all the points of $\mathbf{S}$ and on none of the
points of $\mathbb{Z}^2-\mathbf{S}$. This theorem appears in section
3. The key to its proof is an extension of the theorem of Adleman,
Cheng, Goel, Huang, Kempe, Moisset de Espan\'{e}s, and Rothemund
\cite{ACGHKMR02} on the number of tile types required for a finite
tree to self-assemble from a single seed tile at its root.

Our positive result is that a slight modification of $\mathbf{S}$,
the {\it fibered Sierpinski triangle} $\mathbf{T}$ illustrated in
Figure 2, strictly self-assembles in the Tile Assembly Model.
Intuitively, the fibered Sierpinski triangle $\mathbf{T}$ (defined
precisely in section 4) is constructed by following the recursive
construction of $\mathbf{S}$ but also adding a thin {\it fiber} to
the left and bottom edges of each stage in the construction. These
fibers, which carry data in an algorithmically directed
self-assembly of $\mathbf{T}$, have thicknesses that are logarithmic
in the sizes of the corresponding stages of $\mathbf{T}$. This means
that $\mathbf{T}$ is visually indistinguishable from $\mathbf{S}$ at
sufficiently large scales. Mathematically, it implies that
$\mathbf{T}$ has the same fractal dimension as $\mathbf{S}$.

Since our strict self-assembly must tile the set $\mathbf{T}$ ``from
within,'' the algorithm that directs it is perforce more involved
than the simple XOR algorithm that directs Winfree's
seven-tile-type, non-strict self-assembly of $\mathbf{S}$. Our
algorithm, which is described in section 5, makes extensive,
recursive use of optimal counters \cite{CGM04}, coupled with
measured delay and corner-turning operations.  It uses 51 tile
types, but these are naturally partitioned into small functional
groups, so that we can use Soloveichik and Winfree's local
determinism method \cite{SolWin07} to prove that ${\bf T}$ strictly
self-assembles.

\section{Preliminaries}
\subsection{Notation and Terminology}
We work in the discrete Euclidean plane $\Z^2 = \Z \times \Z$. We
write $U_2$ for the set of all {\it unit vectors}, i.e., vectors of
length $1$, in $\Z^2$.  We regard the four elements of $U_2$ as
(names of the cardinal) {\it directions} in $\Z^2$.

We write $[X]^2$ for the set of all $2$-element subsets of a set
$X$.  All {\it graphs} here are undirected graphs, i.e., ordered
pairs $G = (V, E)$, where $V$ is the set of {\it vertices} and $E
\subseteq [V]^2$ is the set of {\it edges}. A {\it cut} of a graph
$G=(V,E)$ is a partition $C=(C_0,C_1)$ of $V$ into two nonempty,
disjoint subsets $C_0$ and $C_1$.

A {\it binding function} on a graph $G = (V,E)$ is a function
$\beta:E\rightarrow \mathbb{N}$. (Intuitively, if $\{u,v\} \in E$,
then $\beta\left(\{u,v\}\right)$ is the strength with which $u$ is
bound to $v$ by $\{u,v\}$ according to $\beta$. If $\beta$ is a
binding function on a graph $G=(V,E)$ and $C=(C_0,C_1)$ is a cut of
$G$, then the {\it binding strength} of $\beta$ on $C$ is
$$
\beta_C = \left\{ \beta(e) \left| e\in E, e\cap C_0 \ne
\emptyset,\textmd{ and } e \cap C_1\ne \emptyset \right.\right\}.
$$
The {\it binding strength} of $\beta$ on the graph $G$ is then
$$
\beta(G) = \min\left\{ \beta_C \left| C \textmd{ is a cut of } G
\right. \right\}.
$$

A {\it binding graph} is an ordered triple $G=(V,E,\beta)$, where
$(V,E)$ is a graph and $\beta$ is a binding function on $(V,E)$. If
$\tau \in \mathbb{N}$, then a binding graph $G = (V,E,\beta)$ is
$\tau$-{\it stable} if $\beta(V,E)\geq \tau$.

A {\it grid graph} is a graph $G = (V, E)$ in which $V \subseteq
\Z^2$ and every edge $\{\vec{m}, \vec{n} \} \in E$ has the property
that $\vec{m} - \vec{n} \in U_2$.  The {\it full grid graph} on a
set $V \subseteq \Z^2$ is the graph $\fgg{V} = (V, E)$ in which $E$
contains {\it every} $\{\vec{m}, \vec{n} \} \in [V]^2$ such that
$\vec{m} - \vec{n} \in U_2$.

We say that $f$ is a {\it partial function} from a set $X$ to a set
$Y$, and we write $\pfunc{f}{X}{Y}$, if $f: D\rightarrow Y$ for some
set $D \subseteq X$. In this case, $D$ is  the {\it domain} of $f$,
and we write $D = \dom{f}$.

All logarithms here are base-2.

\subsection{The Tile Assembly Model}
We review the basic ideas of the Tile Assembly Model. Our
development largely follows that of \cite{RotWin00,Roth01}, but some
of our terminology and notation are specifically tailored to our
objectives.  In particular, our version of the model only uses
nonnegative ``glue strengths'', and it bestows equal status on
finite and infinite assemblies. We emphasize that the results in
this section have been known for years, e.g., they appear, with
proofs, in \cite{Roth01}.

\begin{definition}
A {\it tile type} over an alphabet $\Sigma$ is a function
$\func{t}{U_2}{\Sigma^* \times \N}$. We write
$t=\pair{\color_t}{\strength_t}$, where
$\func{\color_t}{U_2}{\Sigma^*}$, and $\func{\strength_t}{U_2}{\N}$
are defined by $t(\vec{u}) =
\pair{\color_t(\vec{u})}{\strength_t(\vec{u})}$ for all $\vec{u} \in
U_2$.
\end{definition}

Intuitively, a tile of type $t$ is a unit square. It can be
translated but not rotated, so it has a well-defined ``side
$\vec{u}\;$'' for each $\vec{u} \in U_2$.  Each side $\vec{u}$ of
the tile is covered with a ``glue'' of {\it color}
$\color_t(\vec{u})$ and {\it strength} $\strength_t(\vec{u})$.  If
tiles of types $t$ and $t^\prime$ are placed with their centers at
$\vec{m}$ and $\vec{m}+\vec{u}$, respectively, where $\vec{m} \in
\Z^2$ and $\vec{u} \in U_2$, then they will {\it bind} with strength
$\strength_t(\vec{u}) \cdot \bval{t(\vec{u}) = t^\prime(-\vec{u})}$
where $\bval{\phi}$ is the {\it Boolean} value of the statement
$\phi$. Note that this binding strength is $0$ unless the adjoining
sides have glues of both the same color and the same strength.

For the remainder of this section, unless otherwise specified, $T$
is an arbitrary set of tile types, and $\tau \in \N$ is the
``temperature.''

\begin{definition}
A T-{\it configuration} is a partial function
$\pfunc{\alpha}{\Z^2}{T}$.
\end{definition}

Intuitively, a configuration is an assignment $\alpha$ in which a
tile of type $\alpha(\vec{m})$ has been placed (with its center) at
each point $\vec{m} \in \dom \alpha$. The following data structure
characterizes how these tiles are bound to one another.

\begin{definition}
The {\it binding graph of} a $T$-configuration
$\pfunc{\alpha}{\Z^2}{T}$ is the binding graph $G_\alpha = (V, E,
\beta )$, where $(V, E)$ is the grid graph given by $V =
\dom{\alpha}$, and $\{\vec{m}, \vec{n}\} \in E$ if and only if
\begin{enumerate}
\item $\vec{m} - \vec{n} \in U_n$,
\item $\color_{\alpha(\vec{m})}\left(\vec{n} -
\vec{m}\right) = \color_{\alpha(\vec{n})}\left(\vec{m} -
\vec{n}\right)$, and
\item $\strength_{\alpha(\vec{m})}\left(\vec{n} -\vec{m}\right) >
0$.
\end{enumerate}
%\begin{align*}
%E = \left\{\left. \{\vec{m}, \vec{n}\} \in {[V]}^2 \right| \vec{m} -
%\vec{n} \in U_n,  \color_{\alpha(\vec{m})}\left(\vec{n} -
%\vec{m}\right) = \color_{\alpha(\vec{n})}\left(\vec{m} -
%\vec{n}\right),  \text{and }
%\strength_{\alpha(\vec{m})}\left(\vec{n} -\vec{m}\right) > 0
%\right\},
%\end{align*}
The binding function $\func{\beta}{E}{\Z^+}$ is given by
\[
\beta\left(\{\vec{m}, \vec{n}\}\right) =
\strength_{\alpha(\vec{m})}\left(\vec{n} -\vec{m}\right)
\]
for all $\left\{\vec{m}, \vec{n} \right\} \in E$.
\end{definition}

\begin{definition} \textmd{ }
\begin{enumerate}
  \item [1.]  A $T$-configuration $\alpha$ is $\tau$-{\it stable} if its binding
  graph $G_\alpha$ is $\tau$-stable.
  \item [2.]  A $\tau$-$T$-{\it assembly} is a $T$-configuration that is
  $\tau$-stable.  We write $\asmb^\tau_T$ for the set of all $\tau$-$T$-assemblies.
\end{enumerate}
\end{definition}

\begin{definition}
Let $\alpha$ and $\alpha^\prime$ be $T$-configurations.
\begin{enumerate}
  \item [1.]  $\alpha$ is a {\it subconfiguration} of $\alpha^\prime$, and we
  write $\alpha \sqsubseteq \alpha^\prime$, if $\dom \alpha \subseteq \dom
  \alpha^\prime$ and, for all $\vec{m} \in \dom \alpha$, $\alpha(\vec{m}) =
  \alpha^\prime(\vec{m}).$
  \item [2.]  $\alpha^\prime$ is a {\it single-tile extension} of $\alpha$ if
  $\alpha \sqsubseteq \alpha^\prime$ and $\dom \alpha^\prime - \dom \alpha$ is a
  singleton set.  In this case, we write $\alpha^\prime = \alpha +
  (\ste{\vec{m}}{t})$, where $\{\vec{m}\} = \dom \alpha^\prime - \dom \alpha$
  and $t = \alpha^\prime(\vec{m})$.
\end{enumerate}
\end{definition}

Note that the expression $\alpha + (\ste{\vec{m}}{t})$ is only
defined when $\vec{m} \in \Z^2 - \dom \alpha$.

We next define the ``$\tau$-$t$-frontier'' of a $\tau$-$T$-assembly
$\alpha$ to be the set of all positions at which a tile of type $t$
can be ``$\tau$-stably added'' to the assembly $\alpha$.

\begin{definition}
Let $\alpha \in \asmb^\tau_T$.
\begin{enumerate}
  \item[1.]  For each $t \in T$, the $\tau$-$t$-{\it frontier} of $\alpha$ is
  the set
\[
\frontiertt{\alpha} = \left\{ \vec{m} \in \Z^2-\dom \alpha \left| \;
\sum_{\vec{u} \in U_2} \strength_{t}(\vec{u}) \cdot
\bbval{\alpha(\vec{m} + \vec{u})(-\vec{u}) = t(\vec{u})} \geq \tau
\right. \right\}.
\]
  \item[2.] The $\tau$-{\it frontier} of $\alpha$ is the set
\begin{equation*}
\frontiertau{\alpha} = \bigcup_{t \in T} \frontiertt{\alpha}.
\end{equation*}
\end{enumerate}
\end{definition}

The following lemma shows that the definition of
$\frontiertt{\alpha}$ achieves the desired effect.

\begin{lemma}
Let $\alpha \in \asmbtt$, $\vec{m} \in \mathbb{Z}^2 - \dom \alpha$,
and $t \in T$. Then $\alpha + (\ste{\vec{m}}{t}) \in \asmb^\tau_T$
if and only if $\vec{m} \in \frontiertt{\alpha}$.
\end{lemma}

\begin{notation}
We write $\arrowstett{\alpha}{\alpha^\prime}$ (or, when $\tau$ and
$T$ are clear from context, $\arrowste{\alpha}{\alpha^\prime}$) to
indicate that $\alpha, \alpha^\prime \in \asmbtt$ and
$\alpha^\prime$ is a single-tile extension of $\alpha$.
\end{notation}

In general, self-assembly occurs with tiles adsorbing
nondeterministically and asynchronously to a growing assembly.  We
now define assembly sequences, which are particular ``execution
traces'' of how this might occur.

\begin{definition}
A $\tau$-$T$-{\it assembly sequence} is a sequence
$\vec{\alpha}=(\alpha_i \mid 0 \leq i<k)$ in $\asmbtt$, where $k \in
\mathbb{Z}^+ \cup \{\infty\}$ and, for each $i$ with $1 \leq i+1 <
k$, $\arrowstett{\alpha_i}{\alpha_{i+1}}$.
\end{definition}

Note that assembly sequences may be finite or infinite in length.
Note also that, in any $\tau$-$T$-assembly sequence
$\vec{\alpha}=(\alpha_i \mid 0 \leq i < k)$, we have $\alpha_i
\sqsubseteq \alpha_j$ for all $0 \leq i \leq j < k$.

\begin{definition}
The {\it result} of a $\tau$-$T$-assembly sequence
$\vec{\alpha}=(\alpha_i \mid 0 \leq i < k)$ is the unique
$T$-configuration $\alpha=\res{\vec{\alpha}}$ satisfying $\dom
\alpha = \bigcup_{0 \leq i < k}{\dom \alpha_i}$ and $\alpha_i
\sqsubseteq \alpha$ for each $0 \leq i < k$.
\end{definition}

It is clear that $\res{\vec{\alpha}} \in \asmbtt$ for every
$\tau$-$T$-assembly sequence $\vec{\alpha}$.

\begin{definition}
Let $\alpha, \alpha^\prime \in \asmbtt$.
\begin{enumerate}
\item[1.] A $\tau$-$T$-{\it assembly sequence from} $\alpha$ {\it
to} $\alpha^\prime$ is a $\tau$-$T$-assembly sequence
$\vec{\alpha}=(\alpha_i \mid 0 \leq i < k)$ such that $\alpha_0 =
\alpha$ and $\res{\vec{\alpha}} = \alpha^\prime$.
\item[2.] We write $\arrowtett{\alpha}{\alpha^\prime}$ (or, when
$\tau$ and $T$ are clear from context, $\alpha \To \alpha^\prime$)
to indicate that there exists a $\tau$-$T$-assembly sequence from
$\alpha$ to $\alpha^\prime$.
\end{enumerate}
\end{definition}

A routine dovetailing argument extends the following observation of
\cite{Roth01} to assembly sequences that may have infinite length.

\begin{theorem}
The binary relation $\arrowtett{}{}$ is a partial ordering of
$\asmbtt$.
\end{theorem}

\begin{definition}
An assembly $\alpha \in \asmbtt$ is {\it terminal} if it is a
$\arrowtett{}{}$-maximal element of $\asmbtt$.
\end{definition}

It is clear that an assembly $\alpha$ is terminal if and only if
$\frontiert{\alpha} = \emptyset$.

We now note that every assembly is $\arrowtett{}{}$-bounded by
(i.e., can lead to) a terminal assembly.

\begin{lemma}
\label{terminal} For each $\alpha \in \asmbtt$, there exists
$\alpha^\prime \in \asmbtt$ such that
$\arrowtett{\alpha}{\alpha^\prime}$ and $\alpha^\prime$ is terminal.
\end{lemma}

We now define tile assembly systems.

\begin{definition} \textmd{ }
\begin{enumerate}
\item[1.] A {\it generalized tile assembly system}
({\it GTAS}) is an ordered triple
$$
\mathcal{T} = (T,\sigma,\tau),
$$
where $T$ is a set of tile types, $\sigma \in \asmbtt$ is the {\it
seed assembly}, and $\tau \in \mathbb{N}$ is the {\it temperature}.
\item[2.] A {\it tile assembly system} ({\it TAS}) is a GTAS $\mathcal{T} = (T,\sigma,\tau)$
in which the sets $T$ and $\dom \sigma$ are finite.
\end{enumerate}
\end{definition}

Intuitively, a ``run'' of a GTAS $\mathcal{T}=(T,\sigma,\tau)$ is
any $\tau$-$T$-assembly sequence $\vec{\alpha} = (\alpha_i \mid 0
\leq i < k)$ that begins with $\alpha_0 = \sigma$. Accordingly, we
define the following sets.

\begin{definition} Let $\mathcal{T} = (T,\sigma,\tau)$ be a GTAS.
\begin{enumerate}
\item[1.] The {\it set of assemblies produced by} $\mathcal{T}$ is
$$
\prodasm{T} = \left\{ \alpha \in \asmbtt \left|
\arrowtett{\sigma}{\alpha}
 \right. \right\}.
$$
\item[2.] The {\it set of terminal assemblies produced by}
$\mathcal{T}$ is
$$
\termasm{T} = \left\{ \left. \alpha \in \mathcal{A}[\mathcal{T}]
\right| \alpha\textrm{ is terminal} \right\}.
$$
\end{enumerate}
\end{definition}

\begin{definition}
A GTAS $\mathcal{T} = (T,\sigma,\tau)$ is {\it directed} if the
partial ordering $\arrowtett{}{}$ directs the set
$\prodasm{\mathcal{T}}$, i.e., if for each $\alpha,\alpha^\prime \in
\prodasm{\mathcal{T}}$ there exists $\hat{\alpha} \in
\prodasm{\mathcal{T}}$ such that $\arrowtett{\alpha}{\hat{\alpha}}$
and $\arrowtett{\alpha^\prime}{\hat{\alpha}}$.
\end{definition}

We are using the terminology of the mathematical theory
of relations here.  The reader is cautioned that the term
"directed" has also been used for a different, more
specialized notion in self-assembly \cite{AKKR02}.

Directed tile assembly systems are interesting because they are
precisely those tile assembly systems that produce unique terminal
assemblies.

\begin{theorem}
A GTAS $\mathcal{T}$ is directed if and only if $\left| \termasm{T}
\right|=1$.
\end{theorem}

In the present paper, we are primarily interested in the
self-assembly of sets.

\begin{definition}
Let $\mathcal{T} = (T,\sigma,\tau)$ be a GTAS, and let $X \subseteq
\mathbb{Z}^2$.
\begin{enumerate}
\item[1.] The set $X$ {\it weakly self-assembles} in $\mathcal{T}$
if there is a set $B \subseteq T$ such that, for all $\alpha \in
\termasm{T}$, $\alpha^{-1}(B) = X$.
\item[2.] The set $X$ {\it strictly self-assembles} in $\mathcal{T}$
if, for all $\alpha \in \termasm{T}$, $\dom \alpha = X$.
\end{enumerate}
\end{definition}

Intuitively, a set $X$ weakly self-assembles in $\mathcal{T}$ if
there is a designated set $B$ of ``black'' tile types such that
every terminal assembly of $\mathcal{T}$ ``paints the set $X$ - and
only the set $X$ - black''. In contrast, a set $X$ strictly
self-assembles in $\mathcal{T}$ if every terminal assembly of
$\mathcal{T}$ has tiles on the set $X$ and only on the set $X$.
Clearly, every set that strictly self-assembles in a GTAS
$\mathcal{T}$ also weakly self-assembles in $\mathcal{T}$.

We now have the machinery to say what it means for a set in the
discrete Euclidean plane to self-assemble in either the weak or the
strict sense.

\begin{definition} Let $X \subseteq \mathbb{Z}^2$.
\begin{enumerate}
\item[1.] The set $X$ {\it weakly self-assembles} if there is a TAS $\mathcal{T}$ such that $X$ weakly self-assembles in
$\mathcal{T}$.
\item[2.] The set $X$ {\it strictly self-assembles} if there is a TAS $\mathcal{T}$ such that $X$ strictly self-assembles in
$\mathcal{T}$.
\end{enumerate}
\end{definition}

Note that $\mathcal{T}$ is required to be a TAS, i.e., finite, in
both parts of the above definition.

\subsection{Local Determinism}
The proof of our second main theorem uses the local determinism
method of Soloveichik and Winfree \cite{SolWin07}, which we now
review.

\begin{notation}
For each $T$-configuration $\alpha$, each $\vec{m} \in
\mathbb{Z}^2$, and each $\vec{u} \in U_2$,
$$
\text{str}_{\alpha}(\vec{m},\vec{u}) =
\text{str}_{\alpha(\vec{m})}(\vec{u})\cdot
\bval{\alpha(\vec{m})(\vec{u}) = \alpha(\vec{m}+\vec{u})(-\vec{u})}.
$$
(The Boolean value on the right is 0 if $\{\vec{m},\vec{m}+\vec{u}\}
\nsubseteq \dom{\alpha}$.)
\end{notation}

\begin{notation}
If $\vec{\alpha} = (\alpha_i | 0\leq i<k)$ is a $\tau$-$T$-assembly
sequence and $\vec{m} \in \mathbb{Z}^2$, then the
$\vec{\alpha}$-{\it index} of $\vec{m}$ is
$$
i_{\vec{\alpha}}(\vec{m}) = \min\{ i\in \mathbb{N} \left| \vec{m}
\in \dom{\alpha_i} \right. \}.
$$
\end{notation}

\begin{observation} $\vec{m} \in \dom{\res{\vec{\alpha}}} \Leftrightarrow
i_{\vec{\alpha}}(\vec{m}) < \infty$.
\end{observation}

\begin{notation}
If $\vec{\alpha} = (\alpha_i | 0\leq i<k)$ is a $\tau$-$T$-assembly
sequence, then, for $\vec{m},\vec{m}' \in \mathbb{Z}^2$,
$$
\vec{m} \prec_{\vec{\alpha}} \vec{m}' \Leftrightarrow
i_{\vec{\alpha}}(\vec{m}) < i_{\vec{\alpha}}(\vec{m}').
$$
\end{notation}

\begin{definition}
\label{local_determinism_sets} (Soloveichik and Winfree
\cite{SolWin07}) Let $\vec{\alpha} = (\alpha_i | 0\leq i < k)$ be a
$\tau$-$T$-assembly sequence, and let $\alpha = \res{\vec{\alpha}}$.
For each location $\vec{m} \in \dom{\alpha}$, define the following
sets of directions.
\begin{enumerate}
\item[1.] $\textmd{IN}^{\vec{\alpha}}(\vec{m}) = \left\{ \vec{u} \in U_2 \left| \vec{m}+\vec{u} \prec_{\vec{\alpha}} \vec{m} \textmd{ and }
\textmd{str}_{\alpha_{i_{\vec{\alpha}}(\vec{m})}}(\vec{m},\vec{u})>0
\right.\right\}$.
\item[2.] $\textmd{OUT}^{\vec{\alpha}}(\vec{m}) = \left\{ \vec{u}\in U_2 \left|
 -\vec{u} \in \textmd{IN}^{\vec{\alpha}}(\vec{m}+\vec{u} \right.)
\right\}$.
\end{enumerate}
\end{definition}

Intuitively, $\textmd{IN}^{\vec{\alpha}}(\vec{m})$ is the set of
sides on which the tile at $\vec{m}$ initially binds in the assembly
sequence $\vec{\alpha}$, and $\textmd{OUT}^{\vec{\alpha}}(\vec{m})$
is the set of sides on which this tile propagates information to
future tiles.

Note that $\textmd{IN}^{\vec{\alpha}}(\vec{m}) = \emptyset$ for all
$\vec{m} \in \alpha_0$.

\begin{notation}
If $\vec{\alpha} = (\alpha_i | 0 \leq i < k)$ is a
$\tau$-$T$-assembly sequence, $\alpha = \res{\vec{\alpha}}$, and
$\vec{m} \in \dom{\alpha} - \dom{\alpha_0}$, then
$$
\vec{\alpha} \setminus \vec{m} = \alpha \upharpoonright
\left(\dom{\alpha} - \{\vec{m}\} -
\left(\vec{m}+\textmd{OUT}^{\vec{\alpha}}(\vec{m})\right)\right).
$$
\end{notation}

(Note that $\vec{\alpha} \setminus \vec{m}$ is a $T$-configuration
that may or may not be a $\tau$-$T$-assembly.

\begin{definition}
\label{local_determinism_definition} (Soloveichik and Winfree
\cite{SolWin07}). A $\tau$-$T$-assembly sequence $\vec{\alpha} =
(\alpha_i | 0 \leq i < k)$ with result $\alpha$ is {\it locally
deterministic} if it has the following three properties.
\begin{enumerate}
\item[1.] For all $\vec{m} \in \dom{\alpha} - \dom{\alpha_0}$,
$$
\sum_{\vec{u} \in
\textmd{IN}^{\vec{\alpha}}(\vec{m})}{\textmd{str}_{\alpha_{i_{\vec{\alpha}}(\vec{m})}}(\vec{m},\vec{u})
} = \tau.
$$
\item[2.] For all $\vec{m} \in \dom{\alpha} - \dom{\alpha_0}$ and
all $t \in T-\{\alpha(\vec{m})\}$, $\vec{m} \not \in
\frontiertt{\left(\vec{\alpha} \setminus \vec{m}\right)}$.
\item[3.] $\frontiert{\alpha} = \emptyset$.
\end{enumerate}
\end{definition}

That is, $\vec{\alpha}$ is locally deterministic if (1) each tile
added in $\vec{\alpha}$ ``just barely'' binds to the assembly; (2)
if a tile of type $t_0$ at a location $\vec{m}$ and its immediate
``OUT-neighbors'' are deleted from the {\it result} of
$\vec{\alpha}$, then no tile of type $t \ne t_0$ can attach itself
to the thus-obtained configuration at location $\vec{m}$; and (3)
the result of $\vec{\alpha}$ is terminal.

\begin{definition}
\label{locally_deterministic_tas_def} A GTAS $\mathcal{T} = (T,\sigma,\tau)$ is
{\it locally deterministic} if there exists a locally deterministic
$\tau$-$T$-assembly sequence $\vec{\alpha}=(\alpha_i | 0\leq i < k)$ with
$\alpha_0 = \sigma$.
\end{definition}

\begin{theorem}
\label{local_determinism_theorem} (Soloveichik and Winfree
\cite{SolWin07}) Every locally deterministic \\GTAS is directed.
\end{theorem}

\subsection{Zeta-Dimension}
The most commonly used dimension for discrete fractals is
zeta-dimension, which we use in this paper. The discrete-continuous
correspondence mentioned in the introduction preserves dimension
somewhat generally. Thus, for example, the zeta-dimension of the
discrete Sierpinski triangle is the same as the Hausdorff dimension
of the continuous Sierpinski triangle.

Zeta-dimension has been re-discovered several times by researchers
in various fields over the past few decades, but its origins
actually lie in Euler's (real-valued predecessor of the Riemann)
zeta-function \cite{Euler1737} and Dirichlet series.  For each set
$A \subseteq \Z^2$, define the {\it A-zeta-function} $\zeta_A:[0,
\infty)\rightarrow[0, \infty]$ by $\zeta_A(s) = \sum_{(0, 0) \ne (m,
n) \in A} (|m|+|n|)^{-s}$ for all $s \in [0,\infty)$. Then the {\it
zeta-dimension} of $A$ is
\[
     \textmd{Dim}_\zeta(A) = \inf \{s | \zeta_A(s) < \infty\}.
\]
It is clear that $0 \le \textmd{Dim}_\zeta(A) \le 2$ for all $A
\subseteq \Z^2$. It is also easy to see (and was proven by Cahen in
1894; see also \cite{Apos97,HarWri79}) that zeta-dimension admits
the ``entropy characterization''
\begin{gather*}
\tag*{(2.1)}
     \textmd{Dim}_\zeta(A) = \limsup_{n \rightarrow \infty}\frac{ \log|A_{\le n}|}{\log
     n},
\end{gather*}
where $A_{\le n} = \{(i,j) \in A \mid  |i|+|j| \le n\}$. Various
properties of zeta-dimension, along with extensive historical
citations, appear in the recent paper \cite{ZD}, but our technical
arguments here can be followed without reference to this material.
We use the fact, verifiable by routine calculation, that (2.1) can
be transformed by changes of variable up to exponential, e.g.,
\[
{\rm Dim}_\zeta(A) = \limsup_{n \rightarrow \infty} \frac{ \log
|A_{[0,2^n]}|}{n}
\]
also holds.

\subsection{The Standard Discrete Sierpinski Triangle $\mathbf{S}$}

\begin{figure}[htp]
\begin{center}
\includegraphics[height=4.0in]{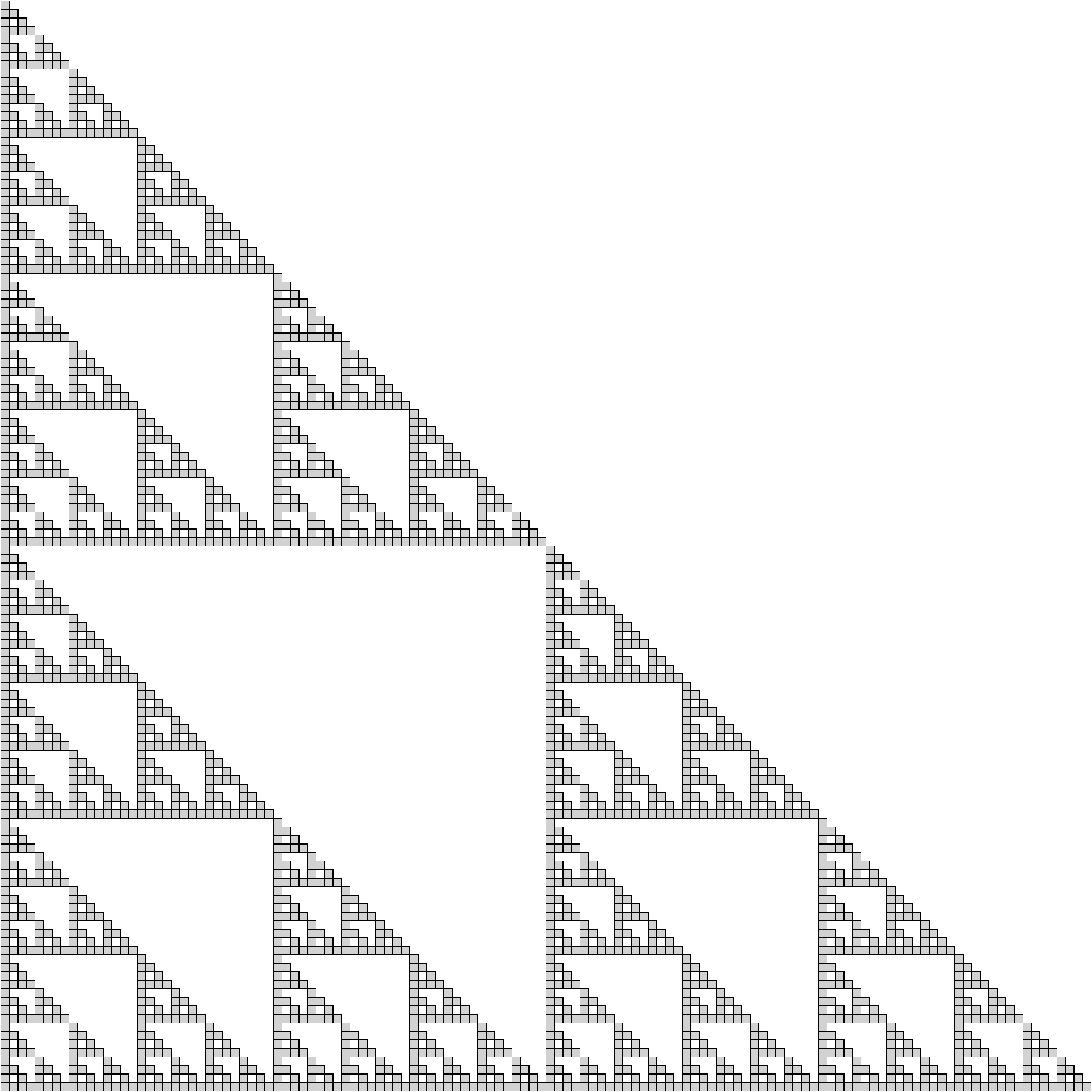}
\end{center}
\caption{The standard discrete Sierpinski triangle $\mathbf{S}.$}
\end{figure}

We briefly review the standard discrete Sierpinski triangle
and the calculation of its zeta-dimension.

Let $V = \{(1, 0), (0, 1) \}$.  Define the sets $S_0, S_1, S_2,
\cdots \subseteq \Z^2$ by the recursion

\begin{gather*}
\tag*{(2.2)}
S_0 = \left\{ \left( 0, 0 \right) \right\}, \\
S_{i+1} = S_i \cup \left( S_i + 2^iV \right),
\end{gather*}
where $A+cB = \{\vec{m} + c\vec{n} | \vec{m} \in A \text{ and } \vec{n} \in B
\}$.  Then the {\it standard discrete Sierpinski triangle} is the  set
\[
{\mathbf S} = \bigcup_{i=0}^{\infty} S_i,
\]
which is illustrated in Figure 1.  It is well known that
$\mathbf{S}$ is the set of all $(k, l) \in \N^2$ such that the
binomial coefficient $\binom{k+l}{k}$ is odd.  For this reason, the
set $\mathbf{S}$ is also called {\it Pascal's triangle modulo 2}. It
is clear from the recursion (2.2) that $|S_i| = 3^i$ for all $i \in
\N$.  The zeta-dimension of $\mathbf{S}$ is thus
\begin{eqnarray*}
{\rm Dim}_\zeta(\mathbf{S}) & = & \limsup_{n \rightarrow \infty} \frac{ \log \left| \mathbf{S}_{[0, 2^n]} \right|}{n} \\
                            & = & \limsup_{n \rightarrow \infty} \frac{ \log \left| S_n \right|}{n} \\
                            & = & \log 3 \\
                            & \approx & 1.585. \\
\end{eqnarray*}

\section{Impossibility of Strict Self-Assembly of ${\rm {\bf{S}}}$}
This section presents our first main theorem, which says that the
standard discrete Sierpinski triangle $\bf{S}$ does not strictly
self-assemble in the Tile Assembly Model. In order to prove this
theorem, we first develop a lower bound on the number of tile types
required for the self-assembly of a set $X$ in terms of the depths
of finite trees that occur in a certain way as subtrees of the full
grid graph $\fgg{X}$ of $X$.

Intuitively, given a set $D$ of vertices of $\fgg{X}$ (which is in
practice the domain of the seed assembly), we now define a
$D$-subtree of $\fgg{X}$ to be any rooted tree in $\fgg{X}$ that
consists of {\it all} vertices of $\fgg{X}$ that lie at or on the
far side of the root from $D$. For simplicity, we state the
definition in an arbitrary graph $G$.

\begin{definition}
Let $G = (V,E)$ be a graph, and let $D \subseteq V$.
\begin{enumerate}
\item[1.] For each $r \in V$, the $D$-$r$-{\it rooted subgraph} of
$G$ is the graph $G_{D,r} = \left( V_{D,r}, E_{D,r} \right)$, where
$$
V_{D,r} = \left\{ v \in V \left| \textmd{ every path from } v
\textmd{ to (any vertex in) } D \textmd{ goes through } r \right.
\right\}
$$
and
$$
E_{D,r} = E \cap \left[ V_{D,r} \right]^2.
$$
(Note that $r \in V_{D,r}$ in any case.)

\item[2.] A $D$-{\it subtree} of $G$ is a rooted tree $B$ with root
$r \in V$ such that $B = G_{D,r}$.

\item[3.] A {\it branch} of a $D$-subtree $B$ of $G$ is a simple
path $\pi = \left( v_0, v_1, \ldots \right)$ in $B$ that starts at
the root of $B$ and either ends at a leaf of $B$ or is infinitely
long.
\end{enumerate}
\end{definition}

We use the following quantity in our lower bound theorem.

\begin{definition}
Let $G = (V,E)$ be a graph, and let $D \subseteq V$. The {\it
finite-tree depth} of $G$ {\it relative to} $D$ is
$$
\ftdepth{D}{G} = \sup\left\{ \textmd{depth(B)} \mid B \textmd{ is a
finite } D \textmd{-subtree of } G\right\}.
$$
\end{definition}

We emphasize that the above supremum is only taken over {\it finite}
$D$-subtrees. It is easy to construct an example in which $G$ has a
$D$-subtree of infinite depth, but $\ftdepth{D}{G} < \infty$.

To prove our lower bound result, we use the following theorem from
\cite{ACGHKMR02}.

\begin{theorem} (Adleman, Cheng, Goel, Huang, Kempe, Moisset de
Espan\'{e}s, and Rothemund \cite{ACGHKMR02})
\label{seven_author_theorem} Let $X \subseteq \mathbb{Z}^2$ with
$|X|<\infty$ be such that $\fgg{X}$ is a tree rooted at the origin.
If $X$ strictly self-assembles in a GTAS $\mathcal{T} =
(T,\sigma,2)$ whose seed $\sigma$ consists of a single tile at the
origin, then $|T| \geq \textrm{depth}\left(\fgg{X}\right)$.
\end{theorem}

Our lower bound result is the following.

\begin{theorem}
\label{lower_bound} Let $X \subseteq \mathbb{Z}^2$. If $X$ strictly
self-assembles in a GTAS $\mathcal{T} = (T,\sigma,\tau)$, then
$$
|T| \geq \ftdepth{\dom{\sigma}}{\fgg{X}}.
$$
\end{theorem}

\begin{proof}
Assume the hypothesis, and let $B$ be a finite
$\dom{\sigma}$-subtree of $\fgg{X}$. If suffices to prove that $|T|
\geq \textmd{depth}(B)$.

Let $\alpha \in \termasm{\mathcal{T}}$, and let $\vec{r}$ be the
root of $B$. Let $\sigma'$ be the assembly with $\dom{\sigma'} =
\{\vec{r}\}$ and $\vec{u}\in U_2$. We define $\sigma'(\vec{r})$ as
follows.
$$
\sigma'(\vec{r})(\vec{u}) = \left\{
\begin{array}{ll}
\left(\color_{\alpha(\vec{r})}(\vec{u}),\strength_{\alpha(\vec{r})}(\vec{u})\right) & \textrm{ if } \vec{r}+\vec{u} \in B \\
(\color_{\alpha(\vec{r})}(\vec{u}),0) & \textrm{ otherwise.}
\end{array} \right.
$$

Then $\mathcal{T}' = (T,\sigma',\tau)$ is a GTAS in which $B$
self-assembles. By Theorem~\ref{seven_author_theorem}, this implies
that $|T| \geq \textmd{depth}(B)$.
\end{proof}

We next show that the standard discrete Sierpinski triangle $\bf{S}$
has infinite finite-tree depth.

\begin{lemma}
\label{finite_set} For every finite set $D \subseteq S$,
$\ftdepth{D}{\fgg{\textbf{S}}} = \infty$.
\end{lemma}

\begin{proof}
Let $D \subseteq \bf{S}$ be finite, and let $m$ be a positive
integer. It suffices to show that $\ftdepth{D}{\fgg{\textbf{S}}}
> m$. Choose $k \in \mathbb{N}$ large enough to satisfy the
following two conditions.
\begin{enumerate}
\item[(i)] $2^k > \max\{a \in \mathbb{N} | (\exists b \in \mathbb{N}) (a,b)\in D
\}$.
\item[(ii)] $2^k > m$.
\end{enumerate}
Let $\vec{r}_k = (2^{k+1},2^k)$, and let
$$
B_k = \left\{ (a,b) \in \textbf{S} \left| a \geq 2^{k+1}, b \geq 2^k
\textrm{ and } a+b\leq 2^{k+2}-1 \right. \right\}.
$$
It is routine to verify that $\fgg{B_k}$ is a finite $D$-subtree of
$\fgg{\bf{S}}$ with root at $\vec{r}$ and depth $2^k$. It follows
that
$$
\ftdepth{D}{\fgg{\textbf{S}}} \geq
\textmd{depth}\left(\fgg{B_k}\right) = 2^k
> m.
$$
\end{proof}

We now have the machinery to prove our first main theorem.

\begin{theorem}
\label{ImpProof} $\bf{S}$ does not strictly self-assemble in the
Tile Assembly Model.
\end{theorem}

\begin{proof}
Let $\mathcal{T} = (T,\sigma,\tau)$ be a GTAS in which $\bf{S}$
strictly self-assembles. It suffices to show that $\mathcal{T}$ is
not a TAS. If $\dom{\sigma}$ is infinite, this is clear, so assume
that $\dom{\sigma}$ is finite. Then Theorem~\ref{lower_bound} and
Lemma~\ref{finite_set} tell us that $|T| = \infty$, whence
$\mathcal{T}$ is not a TAS.
\end{proof}

Before moving on, we note that Theorem \ref{ImpProof} implies the
following lower bound on the number of tile types needed to strictly
assemble any {\it finite} stage $S_n$ of $\bf{S}$.

\begin{corollary}
If a stage $S_n$ of $\bf{S}$ strictly self-assembles in a TAS
$\mathcal{T} = (T,\sigma,\tau)$ in which $\sigma$ consists of a
single tile at the origin, then $|T| \geq 2^n$.
\end{corollary}

If we let $N = \left| S_n \right| = 3^n$, then the above lower bound
exceeds $N^{0.63}$. As Rothemund \cite{Roth01} has noted, a
structure of $N$ tiles that requires $\sqrt{N}$ or more tile types
for its self-assembly cannot be said to feasibly self-assemble.

\section{The Fibered Sierpinski Triangle ${\rm{ {\bf T}}}$}
\label{FIBERED_DEF}

\begin{figure}[htp]
\begin{center}
\includegraphics[height=4.0in]{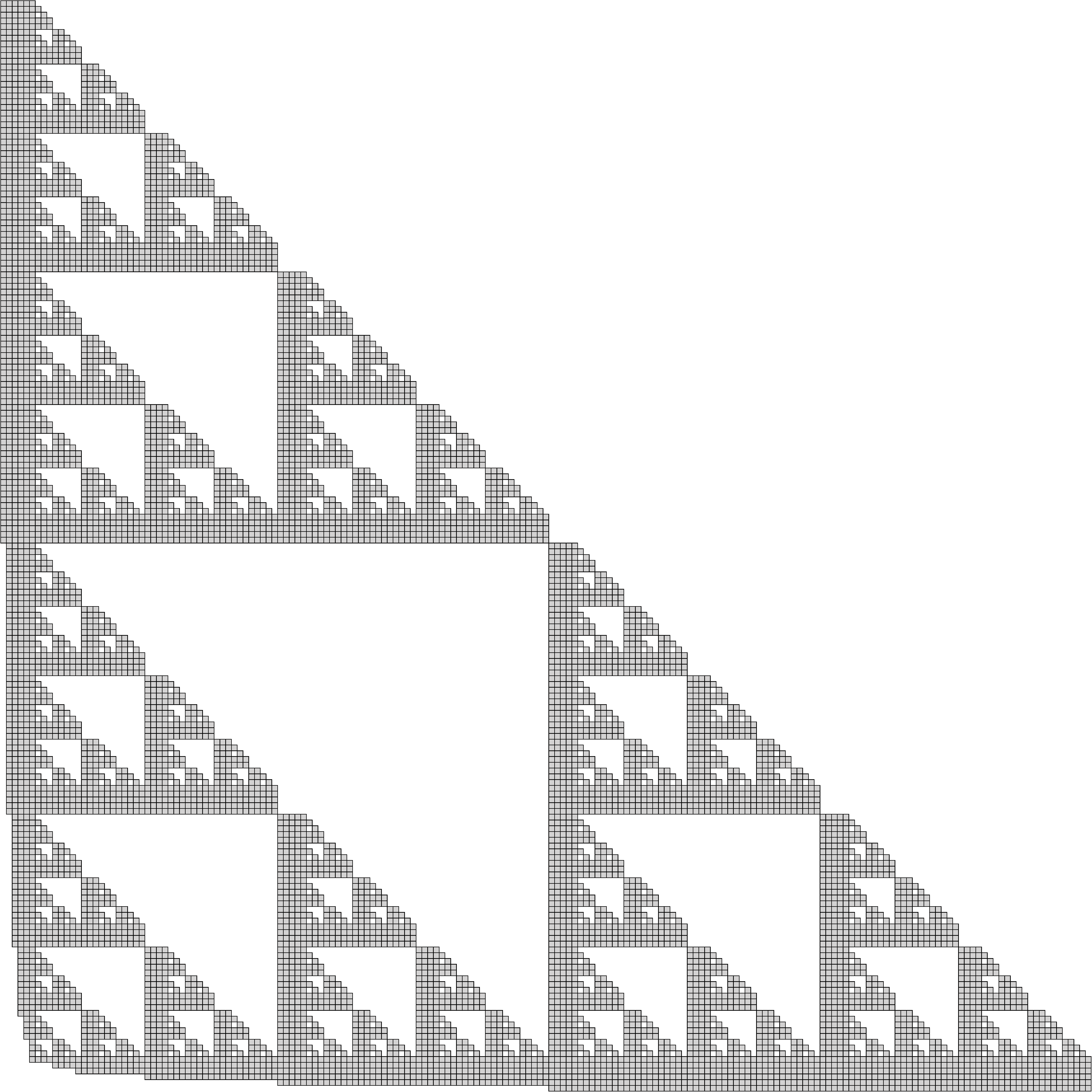}
\end{center}
\caption{\label{fibered_sierpinski} The fibered Sierpinski triangle
$\mathbf{T}.$ }
\end{figure}

We now define the fibered Sierpinski triangle and show that it has
the same zeta-dimension as the standard discrete Sierpinski
triangle.

As in Section~2, let $V = \{(1, 0), (0, 1)\}$.  Our objective is to
define sets of points $T_0, T_1, T_2, \cdots \subseteq \Z^2$, sets
$F_0, F_1, F_2, \cdots \subseteq \Z^2$, and functions $l, f, t : \N
\rightarrow \N$ with the following intuitive meanings.
\begin{enumerate}
\setlength{\itemsep}{6pt}
\item $T_i$ is the $i^{\rm{th}}$ stage of our construction of the fibered
Sierpinski triangle.
\item $F_i$ is the {\it fiber} associated with $T_i$, a thin strip of tiles
along which data moves in the self-assembly process of Section~5.  It is the
smallest set whose union with $T_i$ has a vertical left edge and a horizontal
bottom edge, together with one additional layer added to these two now-straight
edges.
\item $l(i)$ is the length of (number of tiles in) the left (or bottom) edge of
$T_i \cup F_i$.
\item $f(i) = \left| F_i \right|$.
\item $t(i) = \left|T_i\right|$.
\end{enumerate}
These five entities are defined recursively by the equations
\begin{alignat*}{2}
&T_0 = S_2 \text{ (stage 2 in the construction of $S$)}, \\
&F_0 = \left( \left\{ -1 \right\} \times  \left\{ -1, 0, 1, 2, 3 \right\}
\right) \cup \left( \left\{-1, 0, 1, 2, 3 \right\} \times \left\{-1\right\}
\right), \\
&l(0) = 5, \\
&f(0) = 9, \\
&t(0) = 9, \\
&T_{i+1} = T_i \cup \left( \left( T_i \cup F_i \right) +l(i) V \right),
\tag*{(4.1)}\\
&F_{i+1} = F_i \cup \left( \left\{ -i-2 \right\} \times \left\{
-i-2, -i-1,
\cdots , l(i+1)-i-3\right\}\right) \\
&\cup \left( \left\{-i-2, -i-1, \cdots, l(i+1)-i-3 \right\} \times
\left\{-i-2
\right\} \right), \tag*{(4.2)} \\
&l(i+1) = 2l(i) + 1, \\
&f(i+1) = f(i) + 2l(i+1) - 1, \\
&t(i+1) = 3t(i) + 2f(i).
\end{alignat*}
Comparing the recursions (2.1) and (4.1) shows that the sets $T_0,
T_1, T_2, \cdots$ are constructed exactly like the sets $S_0, S_1,
S_2, \cdots$, except that the fibers $F_i$ are inserted into the
construction of the sets $T_i$.  A routine induction verifies that
this recursion achieves conditions 2, 3, 4, and 5 above.  The {\it
fibered Sierpinski triangle} is the set
\begin{gather*}
{\mathbf{T}} = \bigcup_{i = 0}^{\infty} T_i \tag*{(4.3)}
\end{gather*} which is
illustrated in Figure~2. The resemblance between ${\mathbf{S}}$ and
${\mathbf{T}}$ is clear from the illustrations.  We now verify that
${\mathbf{S}}$ and ${\mathbf{T}}$ have the same zeta-dimension.

\begin{lemma} $\textmd{Dim}_{\zeta}(\bf{T}) =
\textmd{Dim}_{\zeta}(\bf{S})$.
\end{lemma}
\begin{proof}
Solving the recurrences for $l$, $f$, and $t$, in that
order, gives the formulas
\begin{alignat*}{2}
&l(i)=3 \cdot 2^{i+1} - 1, \\
&f(i) = 3\left(2^{i+3} - i - 5\right), \\
&t(i) = \frac{3}{2} \left( 3^{i+3} - 2^{i+5} + 2i + 11 \right),
\end{alignat*}
which can be routinely verified by induction.  It follows readily
that \[ {\rm Dim}_\zeta \left(\mathbf{T}\right) = \limsup_{n
\rightarrow \infty} \frac{\log t(n)}{\log l(n)} = \log 3 = {\rm
Dim}_\zeta \left( \mathbf{S} \right). \]
\end{proof}

We note that the thickness $i+1$ of a fiber $F_i$ is $O(\log l(i))$,
i.e., logarithmic in the side length of $T_i$.  Hence the difference
between  $S_i$ and $T_i$ is asymptotically negligible as $i
\rightarrow \infty$. Nevertheless, we show in the next section that
$\mathbf{T}$, unlike $\mathbf{S}$, strictly self-assembles in the
Tile Assembly Model.

\section{Strict Self-Assembly of $\bf{T}$} This section is devoted to
proving our second main theorem, which is the fact that the fibered
Sierpinski triangle $\bf{T}$ strictly self-assembles in the Tile
Assembly Model. Our proof is constructive, i.e., we exhibit a
specific tile assembly system in which $\bf{T}$ strictly
self-assembles.

Our strict self-assembly of $\bf{T}$ is not based directly upon the
recursive definition (4.1). A casual inspection of
Figure~\ref{fibered_sierpinski} suggests that $\bf{T}$ can also be
regarded as a structure consisting of many horizontal and vertical
bars, with each large bar having many smaller bars perpendicular to
it. In subsection 5.1 we give a precise statement and proof of this
``bar characterization'' of $\bf{T}$, which is the basis of our
strict self-assembly. In subsections 5.2 and 5.3 we present the main
functional subsystems of our construction. This gives us a tile
assembly system $\mathcal{T}_{\bf{T}} =
(T_{\mathbf{T}},\sigma_{\mathbf{T}},\tau)$, where
\begin{enumerate}
\item[(i)] the tile set $T_{\mathbf{T}}$ consists of 51 tile types;
\item[(ii)] the seed assembly $\sigma_{\mathbf{T}}$ consists of a single `S' tile
at the origin; and
\item[(iii)] the temperature $\tau$ is 2.
\end{enumerate}
Subsection 5.4 proves that the fibered Sierpinski triangle $\bf{T}$
strictly \\self-assembles in $\mathcal{T}_{\bf{T}}$.

Throughout this section, the temperature $\tau$ is 2. Tiles are
depicted as squares whose various sides are dotted lines, solid
lines, or doubled lines, indicating whether the glue strengths on
these sides are 0, 1, or 2, respectively. Thus, for example, a tile
of the type shown in Figure~3
\begin{figure}[htp]
\begin{center}
\includegraphics[width=0.54in]{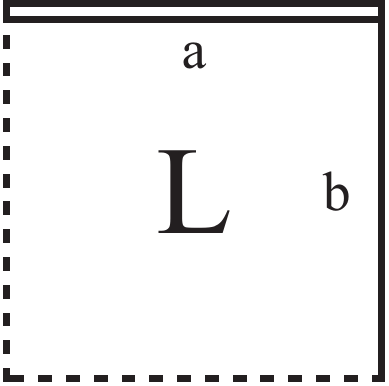}
\label{exampletile} \caption{An example tile type.}
\end{center}
\end{figure}
has glue of strength 0 on the left and bottom, glue of color `a' and
strength 2 on the top, and glue of color `b' and strength 1 on the
right. This tile also has a label `L', which plays no formal role
but may aid our understanding and discussion of the construction.

\subsection{Bar Characterization of $\bf{T}$}
We now formulate the characterization of $\bf{T}$ that guides its
strict self-assembly. At the outset, in the notation of section 4,
we focus on the manner in which the sets $T_i \cup F_i$ can be
constructed from horizontal and vertical bars. Recall that
$$
l(i) = 3\cdot 2^{i+1}-1
$$
is the length of (number of tiles in) the left or bottom edge of
$T_i \cup F_i$.

\begin{definition}
Let $-1 \leq i \in \mathbb{Z}$.
\begin{enumerate}
\item The $S_i$-{\it square} is the set
$$
S_i = \{-i-1,\ldots, 0\}\times\{-i-1,\ldots,0\}.
$$

\item The $X_i$-{\it bar} is the set
$$
X_i = \{1,\ldots,l(i)-i-2\}\times\{-i-1,\ldots,0\}.
$$

\item The $Y_i$-{\it bar} is the set
$$
Y_i = \{-i-1,\ldots,0\}\times\{1,\ldots,l(i)-i-2\}.
$$
\end{enumerate}
\end{definition}

It is clear that the set
$$
S_i \cup X_i\cup Y_i
$$
is the ``outer framework'' of $T_i \cup F_i$. Our attention thus
turns to the manner in which smaller and smaller bars are
recursively attached to this framework.

We use the {\it ruler function}
$$
\rho: \mathbb{Z}^+ \rightarrow \mathbb{N}
$$
defined by the recurrence
\begin{eqnarray*}
\rho(2k+1) & = & 0, \\
\rho(2k) & = & \rho(k) + 1
\end{eqnarray*}
for all $k\in \mathbb{N}$. It is easy to see that $\rho(n)$ is the
(exponent of the) largest power of 2 that divides $n$. Equivalently,
$\rho(n)$ is the number of 0's lying to the right of the rightmost 1
in the binary expansion of $n$ \cite{GrKnPa94}. An easy induction
can be used to establish the following observation.

\begin{observation}
\label{rho_fact} For all $n \in \mathbb{N}$,
$$
\sum_{j=1}^{2^{n+1}-1}{\rho(j)} = 2^{n+1}-n-2.
$$
\end{observation}

Using the ruler function, we define the function
$$
\theta: \mathbb{Z}^+ \rightarrow \mathbb{Z}^+
$$
by the recurrence
\begin{eqnarray*}
\theta(1) & = & 2, \\
\theta(j+1) & = & \theta(j)+\rho(j+1)+2
\end{eqnarray*}
for all $j\in \mathbb{Z}^+$.

We now use the function $\theta$ to define the points at which
smaller bars are attached to the $X_i$- and $Y_i$-bars.

\begin{definition} \textmd{ }
\begin{enumerate}
\item The $j^{\textmd{th}}$ $\theta$-{\it point} of $X_i$ is the
point
$$
\theta_j(X_i) = (\theta(j),1),
$$
lying just above the $X_i$-bar.
\item The $j^{\textmd{th}}$ $\theta$-{\it point} of $Y_i$ is the
point
$$
\theta_j(Y_i) = (1,\theta(j)),
$$
lying just to the right of the $Y_i$-bar.
\end{enumerate}
\end{definition}

The following recursion attaches smaller bars to larger bars in a
recursive fashion.

\begin{definition}
The $\theta$-{\it closures} of the bars $X_i$ and $Y_i$ are the sets
$\theta(X_i)$ and $\theta(Y_i)$ defined by the mutual recursion
\begin{eqnarray*}
\theta\left(X_{-1}\right) & = & X_{-1} \\
\theta\left(Y_{-1}\right) & = & Y_{-1} \\
\theta(X_i) & = & X_i \cup \bigcup_{j=1}^{2^{i+1}-1}{\left(\theta_j(X_i)+\theta\left(Y_{\rho(j)-1}\right)\right)}, \\
\theta(Y_i) & = & Y_i \cup \bigcup_{j=1}^{2^{i+1}-1}{\left(\theta_j(Y_i)+\theta\left(X_{\rho(j)-1}\right)\right)}, \\
\end{eqnarray*}
for all $i\in \mathbb{N}$.
\end{definition}

This definition, along with the symmetry of $\rho$, admit the
following characterizations of $\theta(X_i)$ and $\theta(Y_i)$.

\begin{observation} Let $0\leq i \in \mathbb{N}$.
\begin{enumerate}
\item \begin{eqnarray*}
\theta\left(X_{i+1}\right) & = & \theta(X_i) \cup \left(\left(S_i
\cup \theta(X_i) \cup \theta(Y_i)\right) + (l(i),0)\right) \\
              &  & \quad \cup \left( \{1,\ldots l(i+1)-i-3\} \times \{-i-2\}\right) \\
\end{eqnarray*}

\item \begin{eqnarray*}
\theta\left(Y_{i+1}\right) & = & \theta(Y_i) \cup \left(\left(S_i
\cup \theta(X_i) \cup \theta(Y_i)\right) + (0,l(i))\right) \\
 &  & \quad\cup
\left( \{-i-2\} \times \{ 1,\ldots l(i+1)-i-3\}\right) \\
\end{eqnarray*}
\end{enumerate}
\end{observation}

We have the following characterization of the sets $T_i \cup F_i$.

\begin{lemma}
For all $i\in \mathbb{N}$,
\begin{gather*}
T_i \cup F_i = S_i \cup \theta(X_i) \cup \theta(Y_i).
\end{gather*}
\end{lemma}
\begin{proof}
We proceed by induction on $i$, and note that the case when $i=0$ is
trivial. Assume that, for all $i \in \mathbb{N}$, the lemma holds.
Then we have
\begin{eqnarray*}
T_{i+1} \cup F_{i+1} & \stackrel{\textmd{(4.1)}}{=} & T_i \cup \left(\left(T_i \cup F_i \right) + l(i)V\right) \cup F_{i+1} \\
                     & \stackrel{\textmd{(4.2)}}{=} & \left( T_i \cup F_i \right) \cup \left( \left( T_i \cup F_i \right) + l(i)V\right) \cup \\
                     &   & \quad \left( \{-i-2\} \times \{-i-2,\ldots,l(i+1)-i-3\}\right) \cup \\
                     &   & \quad\quad \left( \{-i-2,\ldots,l(i+1)-i-3\} \times \{-i-2\} \right) \\
                     & \stackrel{\textmd{Ind Hyp}}{=} &  \left( S_i \cup \theta(X_i) \cup \theta(Y_i) \right) \cup \left( \left( S_i \cup \theta(X_i) \cup \theta(Y_i) \right) + l(i)V\right) \cup \\
                     &   & \quad \left( \{-i-2\} \times \{-i-2,\ldots,l(i+1)-i-3\}\right) \cup \\
                     &   & \quad\quad \left( \{-i-2,\ldots,l(i+1)-i-3\} \times \{-i-2\} \right) \\
                     %& = & \left(S_i \cup \left(\{-i-2,\ldots,0\} \times \{-i-2\}\right) \cup \left(\{-i-2\}\times\{-i-2,\ldots,0\}\right)\right) \cup \\
                     %&   & \quad \left(\theta\left(X_i\right) \cup \left(\left(S_i \cup \theta\left(X_i\right) \cup \theta\left(Y_i\right)\right) + (l(i),0)\right)\right) \cup \\
                     %&   & \quad\quad \left(\{1,\ldots, l(i+1)-i-3\}\times\{-i-2\}\right) \cup \\
                     %&   & \quad\quad\quad \left(\theta\left(Y_i\right) \cup \left(\left(S_i \cup \theta\left(X_i\right) \cup \theta\left(Y_i\right)\right) + (0,l(i))\right)\right) \cup \\
                     %&   & \quad\quad\quad\quad \left( \{-i-2\} \times \{1,\ldots,l(i+1)-i-3\}\right) \\
                     & \stackrel{\textmd{Observation 5.2}}{=} & S_{i+1} \cup \theta(X_{i+1}) \cup \theta(Y_{i+1}). \\
\end{eqnarray*}
\end{proof}
We now shift our attention to the global structure of the set
$\bf{T}$.

\begin{definition} \textmd{ }
\begin{enumerate}
\item The $x$-{\it axis} of $\bf{T}$ is the set
$$
X = \{ (m,n) \in \mathbf{T} \; | \; m>0, \textmd{ and } n\leq 0\}.
$$
\item The $y$-{\it axis} of $\mathbf{T}$ is the set
$$
Y = \{ (m,n) \in \mathbf{T} \; | \; m\leq 0, \textmd{ and } n>0 \}.
$$
\end{enumerate}
\end{definition}

Intuitively, the $x$-axis of $\bf{T}$ is the part of $\bf{T}$ that
is a ``gradually thickening bar'' lying on and below the (actual)
$x$-axis in $\mathbb{Z}^2$. (see Figure~\ref{fibered_sierpinski}.)
For technical convenience, we have omitted the origin from this set.
Similar remarks apply to the $y$-axis of $\bf{T}$.

Define the sets
\begin{eqnarray*}
\widetilde{X}_{-1} & = & \{(1,0),(2,0),(3,0) \}, \\
\widetilde{Y}_{-1} & = & \{(0,1),(0,2),(0,3) \}. \\
\end{eqnarray*}
For each $i \in \mathbb{N}$, define the translations
\begin{eqnarray*}
S_i^{\rightarrow} & = & (l(i),0)+S_i, \\
S_i^{\uparrow} & = & (0,l(i))+S_i, \\
X_i^{\rightarrow} & = & (l(i),0)+X_i, \\
Y_i^{\uparrow} & = & (0,l(i))+Y_i
\end{eqnarray*}
of $S_i$, $X_i$, and $Y_i$. It is clear by inspection that $X$ is
the disjoint union of the sets
$$
\widetilde{X}_{-1},S_0^{\rightarrow},X_0^{\rightarrow},S_1^{\rightarrow},X_1^{\rightarrow},S_2^{\rightarrow},X_2^{\rightarrow},\ldots,
$$
which are written in their left-to-right order of position in $X$.
More succinctly, we have the following.

\begin{observation} \textmd{ }
\begin{enumerate}
\item $\displaystyle X = \widetilde{X}_{-1} \cup
\bigcup_{i=0}^{\infty}{\left(S_i^{\rightarrow}\cup
X_i^{\rightarrow}\right)}$.

\item $\displaystyle Y = \widetilde{Y}_{-1} \cup
\bigcup_{i=0}^{\infty}{\left(S_i^{\uparrow} \cup
Y_i^{\uparrow}\right)}$.
\end{enumerate}
\label{observationwhatever} Moreover, both of these are disjoint
unions.
\end{observation}

In light of Observation~\ref{observationwhatever}, it is convenient
to define, for each $-1\leq n \in \mathbb{Z}$, the initial segment
$$
\widetilde{X}_n = \widetilde{X}_{-1} \cup
\bigcup_{i=0}^{n}{\left(S_i^{\rightarrow} \cup
X_i^{\rightarrow}\right)}
$$
of $X$ and the initial segment
$$
\widetilde{Y}_n = \widetilde{Y}_{-1} \cup
\bigcup_{i=0}^{n}{\left(S_i^{\uparrow} \cup Y_i^{\uparrow}\right)}
$$
of $Y$. (Note that this is consistent with earlier usage when $n =
-1$.)

The following definition specifies the manner in which bars are
recursively attached to the $x$- and $y$-axes of $\bf{T}$.

\begin{definition} Let $j \in \mathbb{Z}^+$.
\begin{enumerate}
\item The $j^{\textmd{th}}$ $\theta$-{\it point} of $X$ is the point
$$
\theta_j(X) = (\theta(j),1)
$$
lying just above $X$.
\item The $j^{\textmd{th}}$ $\theta$-{\it point} of $Y$ is the point
$$
\theta_j(Y) = (1,\theta(j))
$$
lying just to the right of $Y$.
\end{enumerate}
\end{definition}

\begin{definition}
For all $-1\leq n \in \mathbb{Z}$, the $\theta$-{\it closures} of
the initial segment of the axes $\widetilde{X}_n$ and
$\widetilde{Y}_n$ are the sets
$$
\theta\left(\widetilde{X}_n\right) = \widetilde{X}_n \cup
\bigcup_{j=1}^{2^{n+2}-1}{\left(\theta_j\left(X\right)+\theta\left(Y_{\rho(j)-1}\right)\right)}
$$
and
$$
\theta\left(\widetilde{Y}_n\right) = \widetilde{Y}_n \cup
\bigcup_{j=1}^{2^{n+2}-1}{\left(\theta_j(Y)+\theta\left(X_{\rho(j)-1}\right)\right)},
$$
respectively.
\end{definition}

The following observation is an immediate consequence of the
previous definition.

\begin{observation} Let $0 \leq n \in \mathbb{N}$.
\begin{enumerate}
\item $\theta\left(\widetilde{X}_n\right) =
\theta\left(\widetilde{X}_{n-1}\right) \cup \left(\left(S_n \cup
\theta\left(X_n\right) \cup
\theta\left(Y_n\right)\right)+(l(n),0)\right)$.
\item $\theta\left(\widetilde{Y}_n\right) =
\theta\left(\widetilde{Y}_{n-1}\right) \cup \left(\left(S_n \cup
\theta\left(X_n\right) \cup
\theta\left(Y_n\right)\right)+(0,l(n))\right)$.
\end{enumerate}
\end{observation}

We have the following characterization of $T_n$.

\begin{lemma} For all $-1\leq n\in\mathbb{Z}$,
\begin{gather*}
T_{n+1} = \{(0,0)\} \cup
\theta\left(\widetilde{X}_n\right)\cup\theta\left(\widetilde{Y}_n\right).
\end{gather*}
\end{lemma}
\begin{proof}
We proceed by induction on $n$. When $n = -1$, it is easy to see
that
$$
\{(0,0)\} \cup
\theta\left(\widetilde{X}_{-1}\right)\cup\theta\left(\widetilde{Y}_{-1}\right)
= T_0.
$$
Now assume that, for all $-1 \leq n \in \mathbb{N}$, the lemma
holds. Then we have
\begin{eqnarray*}
T_{n+2} & \stackrel{\textmd{(4.1)}}{=} & T_{n+1} \cup \left(\left(T_{n+1} \cup F_{n+1}\right)+l(n+1)V\right) \\
        & \stackrel{\textmd{Ind Hyp}}{=} & \left(\{(0,0)\} \cup \theta\left(\widetilde{X}_n\right)\cup\theta\left(\widetilde{Y}_n\right)\right) \cup \\
        &   & \quad \left(\left(T_{n+1} \cup F_{n+1}\right)+l(n+1)V\right) \\
        & \stackrel{\textmd{Lemma 5.3}}{=} & \left(\{(0,0)\} \cup \theta\left(\widetilde{X}_n\right)\cup\theta\left(\widetilde{Y}_n\right)\right) \cup \\
        &   & \quad \left(\left(S_{n+1} \cup \theta\left(X_{n+1}\right) \cup \theta\left(Y_{n+1}\right)\right)+l(n+1)V\right) \\
        %& = & \{(0,0)\} \cup \left(\theta\left(\widetilde{X}_n\right) \cup \left(\left(S_{n+1} \cup \theta\left(X_{n+1}\right) \cup \theta\left(Y_{n+1}\right)\right)+(l(n+1),0)\right)  \right) \\
        %&   & \quad \left(\theta\left(\widetilde{Y}_n\right) \cup \left(\left(S_{n+1} \cup \theta\left(X_{n+1}\right) \cup \theta\left(Y_{n+1}\right)\right)+(0,l(n+1))\right) \right) \\
        & \stackrel{\textmd{Observation 5.5}}{=} & \{(0,0)\} \cup \theta\left(\widetilde{X}_{n+1}\right)\cup\theta\left(\widetilde{Y}_{n+1}\right). \\
\end{eqnarray*}
\end{proof}

\begin{definition}
The $\theta$-{\it closures} of the axes $X$ and $Y$ are the sets
$$
\theta(X) = X \cup
\bigcup_{j=1}^{\infty}{\left(\theta_j(X)+\theta\left(Y_{\rho(j)-1}\right)\right)}
$$
and
$$
\theta(Y) = Y \cup
\bigcup_{j=1}^{\infty}{\left(\theta_j(Y)+\theta\left(X_{\rho(j)-1}\right)\right)},
$$
respectively.
\end{definition}

Figure~3 shows the structure of the $Y$-axis.

\begin{figure}[htp]
\begin{center}
\psfrag{Y0}{$\widetilde{Y}_{-1}$} \psfrag{Y1}{$Y_0^{\uparrow}$}
\psfrag{Y2}{$Y_1^{\uparrow}$} \psfrag{Y3}{$Y_2^{\uparrow}$}
\psfrag{S0}{$S_0^{\uparrow}$} \psfrag{S1}{$S_1^{\uparrow}$}
\psfrag{S2}{$S_2^{\uparrow}$}
\includegraphics[width=1.0in]{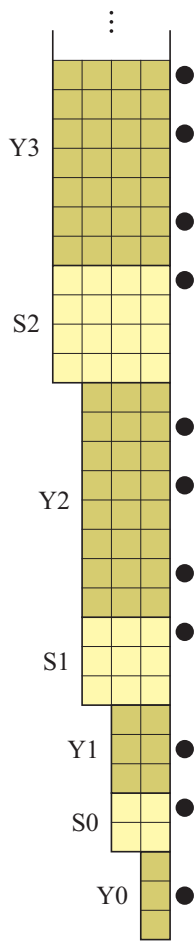}
%\resizebox{1.5in}{!}{\includegraphics{yaxis}}
\label{yaxis} \caption{The structure of $Y$. The dots denote the
$\theta$-points.}
\end{center}
\end{figure}

\begin{lemma} Let $\-1 \leq i \in \mathbb{Z}$.
\begin{enumerate}
\item $\theta(X) =
\displaystyle\bigcup_{i=-1}^{\infty}{\theta\left(\widetilde{X}_i\right)}$.
\item $\theta(Y) =
\displaystyle\bigcup_{i=-1}^{\infty}{\theta\left(\widetilde{Y}_i\right)}$.
\end{enumerate}
\end{lemma}
\begin{proof}
For all $-1 \leq i \in \mathbb{Z}$, it follows from the definition
of $\theta\left(\widetilde{X}_i\right)$, that
\begin{eqnarray*}
\bigcup_{i=-1}^{\infty}{\theta\left(\widetilde{X}_i\right)} & = & \bigcup_{i=-1}^{\infty}{\left(\widetilde{X}_i \cup \bigcup_{j=1}^{2^{i+2}-1}{\left(\theta_j(X)+\theta\left(Y_{\rho(j)-1}\right)\right)}\right)} \\
                                                            & = & \left(\bigcup_{i=-1}^{\infty}{\widetilde{X}_i}\right) \cup \bigcup_{i=-1}^{\infty}{\bigcup_{j=1}^{2^{i+2}-1}{\left(\theta_j(X)+\theta\left(Y_{\rho(j)-1}\right)\right)}} \\
                                                            & \stackrel{\textmd{Observation 5.4}}{=} & X \cup \bigcup_{i=1}^{\infty}{\left(\theta_j(X)+\theta\left(Y_{\rho(j)-1}\right)\right)} \\
                                                            & = & \theta(X). \\
\end{eqnarray*}
The proof of (2) is similar.
\end{proof}

We now have the following characterization of the fibered Sierpinski
triangle.

\begin{theorem}[bar characterization of $\mathbf{T}$]
\label{barchar}
$$
\mathbf{T} = \{(0,0)\} \cup \theta(X)\cup\theta(Y).
$$
\end{theorem}
\begin{proof}
\begin{eqnarray*}
\mathbf{T} & \stackrel{\textmd{(4.3)}}{=} & \bigcup_{i=0}^{\infty}{T_i} \\
           & \stackrel{\textmd{Lemma 5.6}}{=} & \bigcup_{i=-1}^{\infty}{\left(\{(0,0)\} \cup \theta\left(\widetilde{X}_i\right)\cup\theta\left(\widetilde{Y}_i\right)\right)} \\
           & = & \{(0,0)\} \cup \bigcup_{i=-1}^{\infty}{\left( \theta\left(\widetilde{X}_i\right)\cup\theta\left(\widetilde{Y}_i\right) \right)} \\
           & = & \{(0,0)\} \cup \bigcup_{i=-1}^{\infty}{\theta\left(\widetilde{X}_i\right)} \cup \bigcup_{i=-1}^{\infty}{\theta\left(\widetilde{Y}_i\right)} \\
           & \stackrel{\textmd{Lemma 5.7}}{=} & \{(0,0)\} \cup \theta(X) \cup \theta(Y). \\
\end{eqnarray*}
\end{proof}

In the following subsections, we use Theorem~\ref{barchar} to guide
the strict self-assembly of $\bf{T}$.

\subsection{Self-Assembly of the Axes}
In this subsection, we exhibit a TAS in which the $y$-axis of
$\mathbf{T}$ strictly self-assembles. Our tile set is a modification
of the {\it optimal binary counter} (see \cite{CGM04}). If $i+2 \in
\mathbb{N}$ is the width of our modified binary counter, then every
number $1 \leq j \leq 2^{i+1}$ is counted once, and then, if $j \ne
2^{i+1}$, copied $\rho(j)+1$ times. It is easy to verify, using
Observation~\ref{rho_fact}, that this counting scheme produces a
rectangle having a width of $i+2$, and a height of
$$
2^{i+1} + \sum_{j=1}^{2^{i+1}-1}{\left(\rho(j)+1\right)} = l(i)-i-2,
$$
which is precisely the set $Y_i$.

We will now construct our set of tile types $T$.

\begin{construction}
\label{tileset} Let $T$ be the set of 25 tile types shown in
Figure~5.
\begin{figure}[htp]
\begin{center}
\includegraphics[width=2.88in]{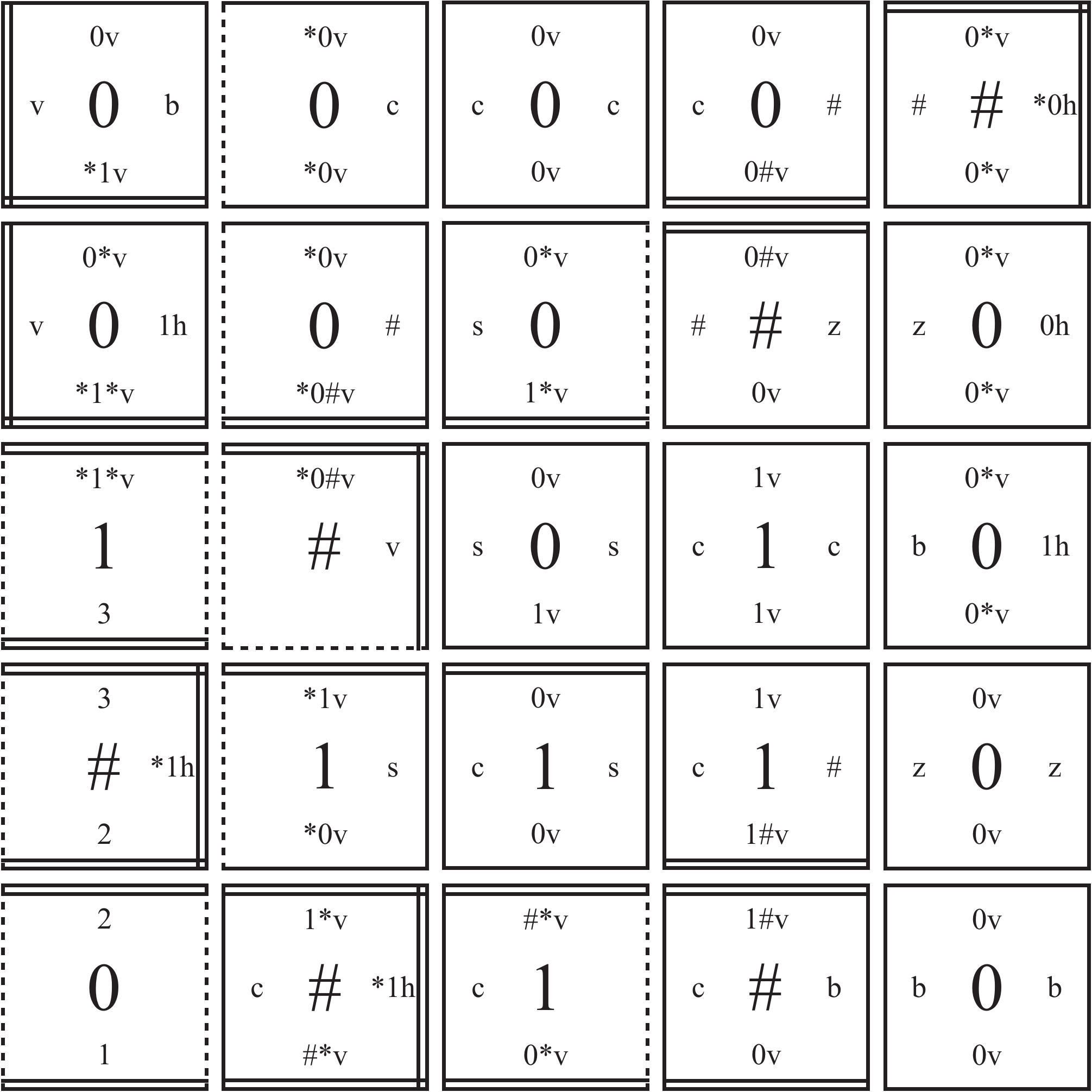}
\label{square_yaxis} \caption{The set of tile types in which the set
$Y$ self-assembles.}
\end{center}
\end{figure}
\end{construction}

The following technical result gives an assembly sequence for the
set $Y_i - (\{-i-1\} \times \{1,\ldots,l(i)-i-2\})$.

\begin{lemma}
\label{techlemma0} Let $n \in \mathbb{N}$. If, for some $m \in
\mathbb{Z}$, $\sigma_n \in \asmbt{\tau}{T}$ satisfies
\begin{enumerate}
\item $\displaystyle\dom{\sigma_n} = \left( \{-n-1,\ldots,0\} \times \{m-1,m\} \right)
$,
\item for all $\vec{v} \in \left(\{ -n-1,\ldots,0\} \times \{m \}\right)$,
$$
\sigma_n\left(\vec{v}\right)(0,1) = \left\{
\begin{array}{ll}
\left(\textmd{0*v},2\right) & \textmd{if } \vec{v} = (0,m)  \\
\left(\textmd{0v},1\right)  & \textmd{if } \vec{v} \ne (0,m) \textmd{ and } \vec{m} \ne (-n-1,m), \\
\end{array}\right.
$$
\end{enumerate}
then there is a $\tau$-$T$-assembly sequence $\vec{\alpha}_n =
\left( \alpha_{i_n} \left| \; 0\leq i_n < k_n < \infty \right.
\right)$, with $\alpha = \res{\vec{\alpha}_n}$, satisfying
\begin{enumerate}
\item $\alpha_{0_n} = \sigma_n$,
\item $\dom{\alpha} = \dom{\sigma_n} \cup \left( \left(Y_n - \left( \{ -n-1 \} \times \{ 1,\ldots,l(n)-n-2 \}\right)\right)+(0,m)\right)$,
\item for all $\vec{v} \in \dom{\alpha} - \dom{\sigma_n}$,
$$
\sum_{\vec{u} \in
\textmd{IN}^{\vec{\alpha}}(\vec{v})}{\textmd{str}_{\alpha_{i_{\vec{\alpha}}(\vec{v})}}(\vec{v},\vec{u})
} = \tau,
$$
\item for all $\vec{v} \in \dom{\alpha} - \dom{\sigma_n}$ and all $t \in
T-\{\alpha(\vec{v})\}$, $\vec{v} \not \in
\frontiertt{\left(\vec{\alpha} \setminus \vec{v}\right)}$,
\item for all $\vec{v} \in \dom{\alpha} - \dom{\sigma_n}$,
$$
\alpha\left(\vec{v}\right)(1,0) = \left\{
\begin{array}{ll}
\left(\textmd{*1h},2\right) & \textmd{if } \exists j \in \mathbb{N}, \textmd{ and }  \vec{v} = (0,\theta(2j+1)+m)  \\
\left(\textmd{*0h},2\right) & \textmd{if } \exists j \in \mathbb{N}, \textmd{ and } \vec{v} = (0,\theta(2j)+m) \\
\left(\textmd{0h},1\right) & \textmd{if } \exists j \in \mathbb{N}, \textmd{ and } \\
                           & \vec{v} \in (\left(\{0\} \times \{\theta\left(j\right)-\rho\left(j\right)+1, \ldots, \theta\left(j\right)-1\}\right) \\
                           & \quad+(0,m))  \\
\left(\textmd{1h},1\right) & \textmd{if } \exists j \in \mathbb{N}, \textmd{ and }  \vec{v} = (0,\theta(j)-\rho(j)+m)  \\
\left(\lambda,0\right) & \textmd{otherwise,} \\
\end{array}\right.
$$
and
\item for all $\vec{v} \in \left\{ (x,y) \in \dom{\alpha} \left| \; (x,y+1) \not \in \dom{\alpha} \right.\right\}$,
$$
\alpha\left(\vec{v}\right)(0,1) = \left\{
\begin{array}{ll}
\left(\textmd{0*v},1\right) & \textmd{if } \vec{v}+(1,0) \not \in \dom{\alpha}  \\
\left(\textmd{0v},1\right)  & \textmd{otherwise.}
\end{array}\right.
$$
\end{enumerate}
\end{lemma}
\begin{proof}
We proceed by induction on $n$, noting that the the base case is
verified in Figure~6.
\begin{figure}[htp]
\begin{center}
\includegraphics[width=0.54in]{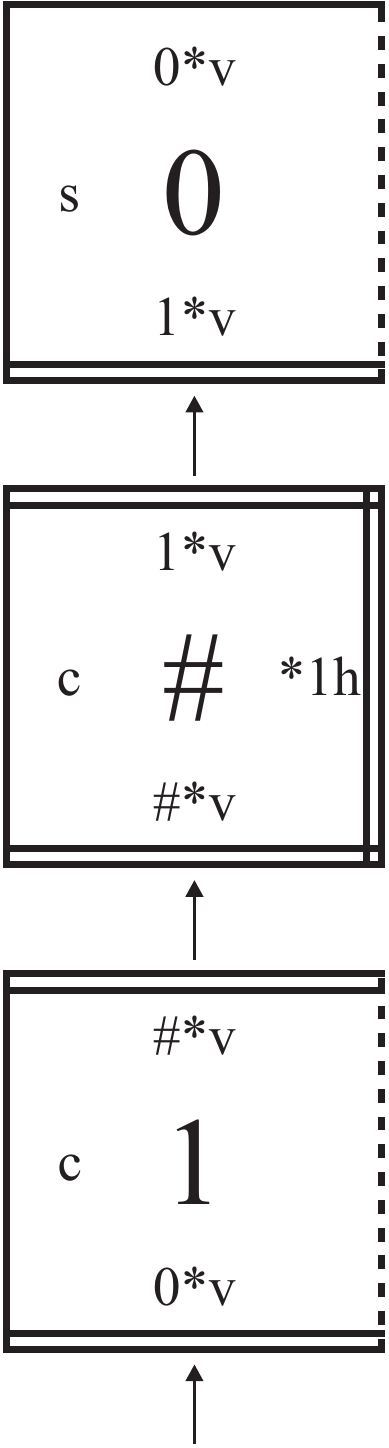}
\label{square_yaxis} \caption{The base case for
Lemma~\ref{techlemma0}.}
\end{center}
\end{figure}
Now assume that the claim holds for all $n \in \mathbb{N}$, and let
$\sigma_{n+1} \in \asmbt{\tau}{T}$ satisfy conditions (1) and (2) of
the hypothesis, taking $m=0$. Let $\sigma_n \in \asmbt{\tau}{T}$
satisfy, for all $\vec{v} \in \mathbb{Z}^2$,
$$
\sigma_n(\vec{v}) = \left\{
\begin{array}{ll}
\sigma_{n+1}(\vec{v}) & \textmd{ if } \vec{v} \in (\{-n-1,\ldots,0\}\times\{-1,0\}) \\
\uparrow & \textmd{ otherwise.} \\
\end{array}
\right.
$$
Then the induction hypothesis tells us that there is an assembly
sequence $\vec{\alpha}_n$, with $\alpha = \res{\vec{\alpha}_n}$,
satisfying conditions (1), (2), (3), (4), (5), and (6) of the
conclusion, taking $m=0$. Define the assembly sequence $$
\vec{\alpha}_{n+1} = \left( \left. \alpha_{i_{n+1}} \; \right| \; 0
\leq i_{n+1} < k_{n+1} < \infty \right)$$ satisfying
\begin{enumerate}
\item $\alpha_{0_{n+1}} = \sigma_{n+1}$,
\item for all $0 \leq i_{n+1} < k_n$, $\alpha_{i_{n+1}} =
\alpha_{i_n}$,
\item for all $0 \leq i < l(n)-n-3$, $\alpha_{\left(k_n+i\right)_{n+1}} =
\alpha_{\left(k_n+i-1\right)_{n+1}} + ((-n-1,i+1) \mapsto t)$, where
$t$ is the tile type shown in Figure~7, and
\begin{figure}[htp]
\begin{center}
\includegraphics[width=0.54in]{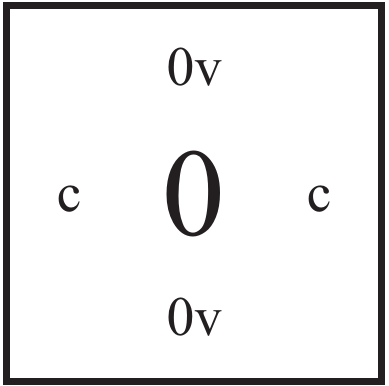}
\label{t} \caption{The tile type $t$.}
\end{center}
\end{figure}
\item $\displaystyle\alpha_{\left(k_n+l(n)-n-3\right)_{n+1}} =
\alpha_{\left(k_n+l(n)-n-4\right)_{n+1}} + ((-n-1,l(n)-n-2) \mapsto
t')$, where $t'$ is the tile type shown in Figure~8.
\begin{figure}[htp]
\begin{center}
\includegraphics[width=0.54in]{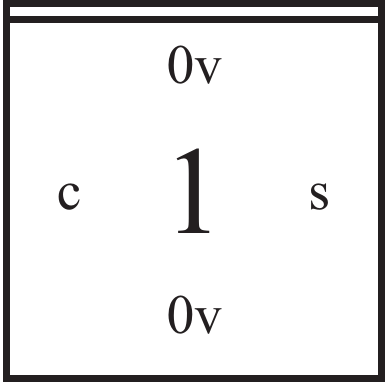}
\label{tt} \caption{The tile type $t'$.}
\end{center}
\end{figure}
\end{enumerate}
Notice that for all $$\vec{v} \in \left\{\left. (x,y) \in
\dom{\alpha_{\left(k_n+l(n)-n-3\right)_{n+1}}} \; \right| \; (x,y+1)
\not \in \dom{\alpha_{\left(k_n+l(n)-n-3\right)_{n+1}}} \right\}$$,
$$
\alpha_{\left(k_n+l(n)-n-3\right)_{n+1}}\left(\vec{v}\right)(0,1) =
\left\{
\begin{array}{ll}
\left(\textmd{0v},2\right) & \textmd{if } \vec{v}+(-1,0) \not \in \dom{\alpha}  \\
\left(\textmd{0*v},1\right) & \textmd{if } \vec{v}+(1,0) \not \in \dom{\alpha}  \\
\left(\textmd{0v},1\right)  & \textmd{otherwise.}
\end{array}\right.
$$
\begin{figure}[htp]
\begin{center}
\includegraphics[width=2.3in]{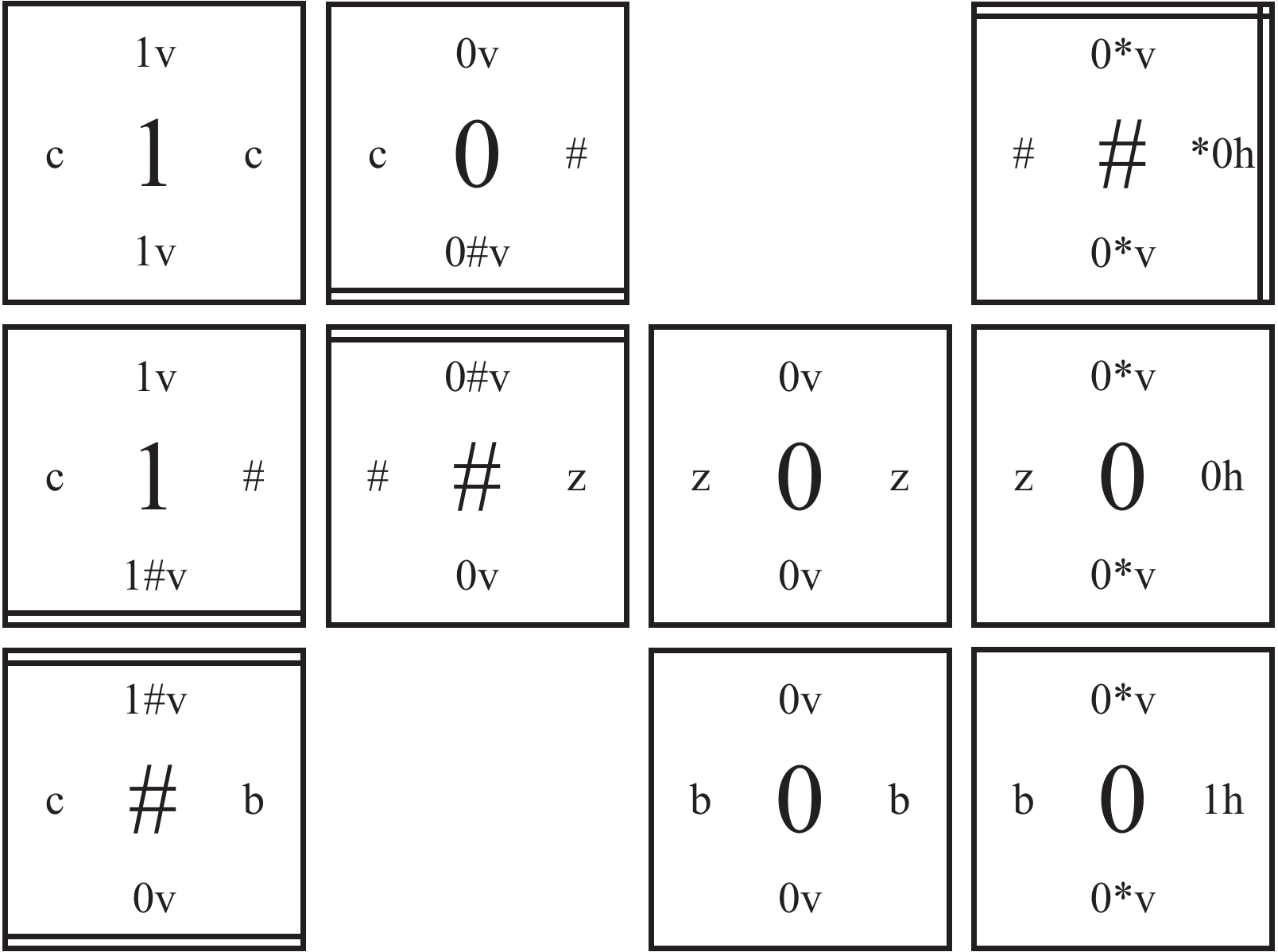}
\label{squaretiletypes} \caption{Tile types in which
$S_n^{\uparrow}$ self-assembles.}
\end{center}
\end{figure}
The tile types shown in Figure~9 testify that there is a
$\tau$-$T$-assembly sequence
$$
\vec{\Box} = \left( \alpha_{\left(k_n+l(n)-n-2\right)_{n+1}},
\ldots, \alpha_{\left(k_n+l(n)-n-2+(n+1)^2-1\right)_{n+1}}\right)
$$
having the property that
$$\dom{\alpha_{\left(k_n+l(n)-n-2+(n+1)^2-1\right)_{n+1}}} -
\dom{\alpha_{\left(k_n+l(n)-n-2\right)_{n+1}}} = S_n^{\uparrow},$$
and for all $\vec{v} \in (\{-n-1,\ldots 0\} \times \{l(n)\})$,
$$
\alpha_{\left(k_n+l(n)-n-2+(n+1)^2-1\right)_{n+1}}\left(\vec{v}\right)(0,1)
= \left\{
\begin{array}{ll}
\left(\textmd{0*v},2\right) & \textmd{if } \vec{v} = (0,l(n))  \\
\left(\textmd{1v},1\right) & \textmd{if } \vec{v} = (-n-1,l(n))  \\
\left(\textmd{0v},1\right)  & \textmd{otherwise.} \\
\end{array}\right.
$$
Let $\sigma'_n \in \asmbt{\tau}{T}$ satisfy, for all $\vec{v} \in
\mathbb{Z}^2$,
$$
\sigma'_n(\vec{v}) = \left\{
\begin{array}{ll}
\alpha_{\left(k_n+l(n)-n-2+(n+1)^2-1\right)_{n+1}}(\vec{v}) & \textmd{ if } \vec{v} \in (\{-n-1,\ldots,0\} \\
                                                            & \quad\times\{l(n)-1,l(n)\}) \\
\uparrow & \textmd{ otherwise.} \\
\end{array}
\right.
$$
Once again, we appeal to the induction hypothesis, which tells us
that there is an assembly sequence $\vec{\alpha}'_n$, with $\alpha'
= \res{\vec{\alpha}'_n}$, satisfying, with $m=l(n)$, conditions (1),
(2), (3), (4), (5), and (6) of the conclusion. Thus, we can define
an assembly sequence $\vec{\alpha}'_{n+1} = \left( \left.
\alpha_{i'_{n+1}} \; \right| \; 0 \leq i'_{n+1} < k'_{n+1} < \infty
\right)$ satisfying
\begin{enumerate}
\item $\alpha_{0'_{n+1}} =
\alpha_{\left(k_n+l(n)-n-2+(n+1)^2-1\right)_{n+1}}$,
\item for all $1 \leq i'_{n+1} < k'_n$, $\alpha_{i'_{n+1}} =
\alpha_{i'_n}$,
\item for all $0 \leq i < l(n)-n-2$, $$\alpha_{\left(k'_n+i\right)_{n+1}} =
\alpha_{\left(k'_n+i-1\right)_{n+1}} + \left((-n-1,i+1+l(n)) \mapsto
t''\right),$$ where $t''$ is the tile type shown in Figure~10, and
\begin{figure}[htp]
\begin{center}
\includegraphics[width=0.54in]{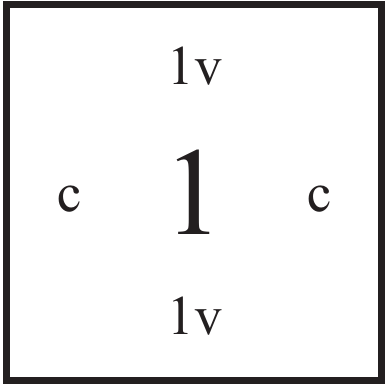}
\label{ttt} \caption{The tile type $t''$.}
\end{center}
\end{figure}
\item $\displaystyle\alpha_{\left(k'_n+l(n)-n-3\right)_{n+1}} =
\alpha_{\left(k'_n+l(n)-n-4\right)_{n+1}} + \left((-n-1,l(n+1)-n-3)
\mapsto t'''\right)$, where $t'''$ is the tile type shown in
Figure~11.
\begin{figure}[htp]
\begin{center}
\includegraphics[width=0.54in]{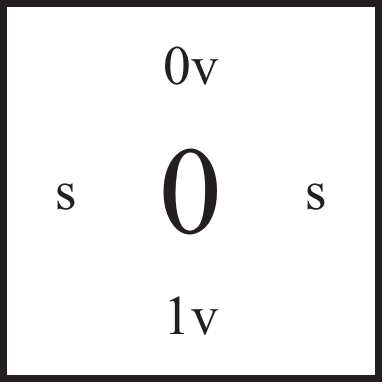}
\label{tttt} \caption{The tile type $t'''$.}
\end{center}
\end{figure}
\end{enumerate}
It is routine to verify that $\vec{\alpha} \cup \vec{\Box} \cup
\vec{\alpha}'$ is a $\tau$-$T$-assembly sequence satisfying
conditions (1), (2), (3), (4), (5), and (6) of the conclusion.
\end{proof}

In the following, we assume the presence of $\widetilde{Y}_{n-1}
\cup S_n^{\uparrow}$, and use Lemma~\ref{techlemma0} to give an
assembly sequence for the set $Y_i^{\uparrow}$.

\begin{lemma}
\label{techlemma1} Let $n \in \mathbb{N}$. If $\sigma_n \in
\asmbt{\tau}{T}$ satisfies
\begin{enumerate}
\item $\displaystyle\dom{\sigma_n} = \left(\widetilde{Y}_{-1} \cup S_0^{\uparrow}\right) \cup
\bigcup_{i=1}^{n}{\left( Y_{i-1}^{\uparrow} \cup S_i^{\uparrow}
\right)}$, and
\item for all $\vec{m} \in \left\{ (x,y) \in \dom{\sigma_n} \left| \; (x,y+1) \not \in \dom{\sigma_n} \right.\right\}$,
$$
\sigma_n\left(\vec{m}\right)(0,1) = \left\{
\begin{array}{ll}
\left(\textmd{*0v},1\right) & \textmd{if } \vec{m}+(-1,0) \not \in \dom{\sigma_n} \\
\left(\textmd{0*v},2\right) & \textmd{if } \vec{m}+(1,0) \not \in \dom{\sigma_n}  \\
\left(\textmd{0v},1\right)  & \textmd{otherwise},
\end{array}\right.
$$
\end{enumerate}
then there is a $\tau$-$T$-assembly sequence $\vec{\alpha} = \left(
\alpha_i \left| \; 0\leq i < k < \infty \right. \right)$, with
$\alpha = \res{\vec{\alpha}}$, satisfying
\begin{enumerate}
\item $\alpha_0 = \sigma_n$,
\item $\dom{\alpha} = \widetilde{Y}_n$,
\item for all $\vec{m} \in \dom{\alpha} - \dom{\sigma_n}$,
$$
\sum_{\vec{u} \in
\textmd{IN}^{\vec{\alpha}}(\vec{m})}{\textmd{str}_{\alpha_{i_{\vec{\alpha}}(\vec{m})}}(\vec{m},\vec{u})
} = \tau,
$$
\item for all $\vec{m} \in \dom{\alpha} - \dom{\sigma_n}$ and all $t \in
T-\{\alpha(\vec{m})\}$, $\vec{m} \not \in
\frontiertt{\left(\vec{\alpha} \setminus \vec{m}\right)}$,
\item for all $\vec{m} \in \dom{\alpha} - \dom{\sigma_n}$,
$$
\alpha\left(\vec{m}\right)(1,0) = \left\{
\begin{array}{ll}
\left(\textmd{*1h},2\right) & \textmd{if } \exists j \in \mathbb{N}, \textmd{ and }  \vec{m} = (0,\theta(2j+1))  \\
\left(\textmd{*0h},2\right) & \textmd{if } \exists j \in \mathbb{N}, \textmd{ and } \vec{m} = (0,\theta(2j)) \\
\left(\textmd{0h},1\right) & \textmd{if } \exists j \in \mathbb{N}, \textmd{ and }  \\
                           & \quad\vec{m} \in \left(\{0\} \times \{\theta\left(j\right)-\rho\left(j\right)+1, \ldots, \theta\left(j\right)-1\}\right)  \\
\left(\textmd{1h},1\right) & \textmd{if } \exists j \in \mathbb{N}, \textmd{ and }  \vec{m} = (0,\theta(j)-\rho(j))  \\
\left(\lambda,0\right) & \textmd{otherwise,} \\
\end{array}\right.
$$
and
\item for all $\vec{m} \in \left\{ (x,y) \in \dom{\alpha} \left| \; (x,y+1) \not \in \dom{\alpha} \right.\right\}$,
$$
\alpha\left(\vec{m}\right)(0,1) = \left\{
\begin{array}{ll}
\left(\textmd{*1v},2\right) & \textmd{if } \vec{m}+(-1,0) \not \in \dom{\alpha} \\
\left(\textmd{0*v},1\right) & \textmd{if } \vec{m}+(1,0) \not \in \dom{\alpha}  \\
\left(\textmd{0v},1\right)  & \textmd{otherwise.}
\end{array}\right.
$$
\end{enumerate}
\end{lemma}
\begin{proof}
Assume the hypothesis. Then, with an appropriate choice of $m \in
\mathbb{Z}$, Lemma~\ref{techlemma0} tells us that there is a
$\tau$-$T$-assembly sequence $\vec{\alpha}_n = ( \alpha_{i_n} \; |
\; 0 \leq i_n < k_n < \infty)$, with $\alpha_n =
\res{\vec{\alpha}_n}$, satisfying $\dom{\alpha_n} = Y_n^{\uparrow} -
(\{-n-1\} \times \{1,\ldots,l(n)-n-2\})$. Define the assembly
sequence $$\vec{\alpha} = \left( \left. \alpha_{i} \; \right| \; 0
\leq i \leq l(n)-n-2 \right)$$ with, $\alpha_0 = \alpha_{0_n}$, and
for all $1 \leq i < l(n)-n-2$,
$$
\alpha_{i} = \alpha_{i-1} + ((-n-1,i) \mapsto t),
$$
where $t$ is the tile type shown in Figure~12, and
\begin{figure}[htp]
\begin{center}
\includegraphics[width=0.54in]{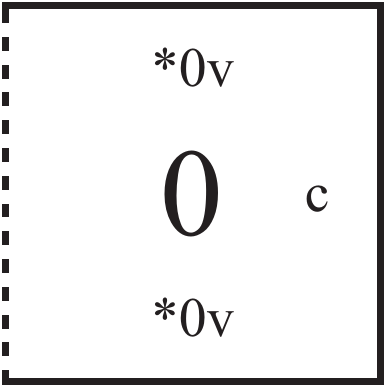}
\label{ttttt} \caption{The tile type $t$.}
\end{center}
\end{figure}
$$
\alpha_{k_n+l(n)-n-2} = \alpha_{k_n+l(n)-n-3} + ((-n-1,l(n)-n-2)
\mapsto t'),
$$
where $t'$ is the tile type shown in Figure~13.
\begin{figure}[htp]
\begin{center}
\includegraphics[width=0.54in]{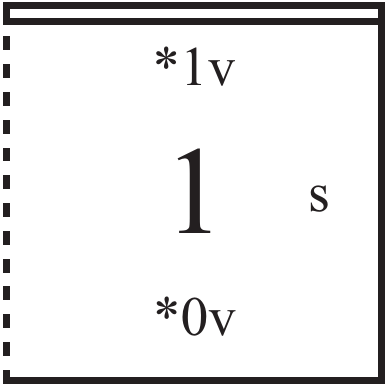}
\label{tttttt} \caption{The tile type $t'$.}
\end{center}
\end{figure}
\end{proof}

Now we assume the presence of the set $\widetilde{Y}_n$ and give an
assembly sequence for $S_{n+1}^{\uparrow}$.

\begin{lemma}
\label{techlemma2} Let $-1 \leq n \in \mathbb{Z}$. If $\sigma_n \in
\asmbt{\tau}{T}$ satisfies
\begin{enumerate}
\item $\displaystyle\dom{\sigma_n} = \widetilde{Y}_{-1} \cup
\bigcup_{i=0}^{n}{\left(S_i^{\uparrow} \cup Y_i^{\uparrow}\right)}$,
and
\item for all $\vec{m} \in \left\{ (x,y) \in \dom{\sigma_n} \left| \; (x,y+1) \not \in \dom{\sigma_n} \right.\right\}$,
$$
\sigma_n\left(\vec{m}\right)(0,1) = \left\{
\begin{array}{ll}
\left(\textmd{*0v},2\right) & \textmd{if } n \geq 0 \textmd{ and } \vec{m}+(-1,0) \not \in \dom{\sigma_n} \\
\left(\textmd{0*v},1\right) & \textmd{if }  n \geq 0 \textmd{ and } \vec{m}+(1,0) \not \in \dom{\sigma_n}  \\
\left(\textmd{*1*v},2\right) & \textmd{if } n = -1  \\
\left(\textmd{0v},1\right)  & \textmd{otherwise},
\end{array}\right.
$$
\end{enumerate}
then there is a $\tau$-$T$-assembly sequence $\vec{\alpha} = \left(
\alpha_i \left| \; 0 \leq i < k < \infty \right. \right)$, with
$\alpha = \res{\vec{\alpha}}$, satisfying
\begin{enumerate}
\item $\alpha_0 = \sigma_n$,
\item $\dom{\alpha} = \widetilde{Y}_n \cup
S_{n+1}^{\uparrow}$,
\item for all $\vec{m} \in \dom{\alpha} - \dom{\sigma_n}$,
$$
\sum_{\vec{u} \in
\textmd{IN}^{\vec{\alpha}}(\vec{m})}{\textmd{str}_{\alpha_{i_{\vec{\alpha}}(\vec{m})}}(\vec{m},\vec{u})
} = \tau,
$$
\item for all $\vec{m} \in \dom{\alpha} - \dom{\sigma_n}$ and
all $t \in T-\{\alpha(\vec{m})\}$, $\vec{m} \not \in
\frontiertt{\left(\vec{\alpha} \setminus \vec{m}\right)}$,
\item for all $\vec{m} \in \left\{ (x,y) \in \dom{\alpha} - \dom{\sigma_n}
\left| \; (x+1,y) \not \in \dom{\alpha} - \dom{\sigma_n}
\right.\right\}$,
$$
\alpha\left(\vec{m}\right)(1,0) = \left\{
\begin{array}{ll}
\left(\textmd{*0h},2\right) & \textmd{if } \vec{m} + (0,1) \not \in \dom{\alpha} - \dom{\sigma_n} \\
\left(\textmd{1h},1\right) & \textmd{if } \vec{m} + (0,-1) \not \in \dom{\alpha} - \dom{\sigma_n} \\
\left(\textmd{0h},1\right) & \textmd{otherwise,}  \\
\end{array}\right.
$$
and

\item for all $\vec{m} \in \left\{ (x,y) \in
\dom{\alpha} \left| \; (x,y+1) \not \in \dom{\alpha}
\right.\right\}$,
$$
\alpha\left(\vec{m}\right)(0,1) = \left\{
\begin{array}{ll}
\left(\textmd{*0v},1\right) & \textmd{if } \vec{m}+(-1,0) \not \in \dom{\alpha} \\
\left(\textmd{0*v},2\right) & \textmd{if } \vec{m}+(1,0) \not \in \dom{\alpha}  \\
\left(\textmd{0v},1\right)  & \textmd{otherwise.}
\end{array}\right.
$$
\end{enumerate}
\end{lemma}
\begin{proof}
This is obvious, and therefore, we omit a detailed proof. See
Figure~14 for an example of the self-assembly of $S_2^{\uparrow}$
``on top of'' $\widetilde{Y}_1$.
\end{proof}

\begin{figure}[htp]
\begin{center}
\includegraphics[width=2.63in]{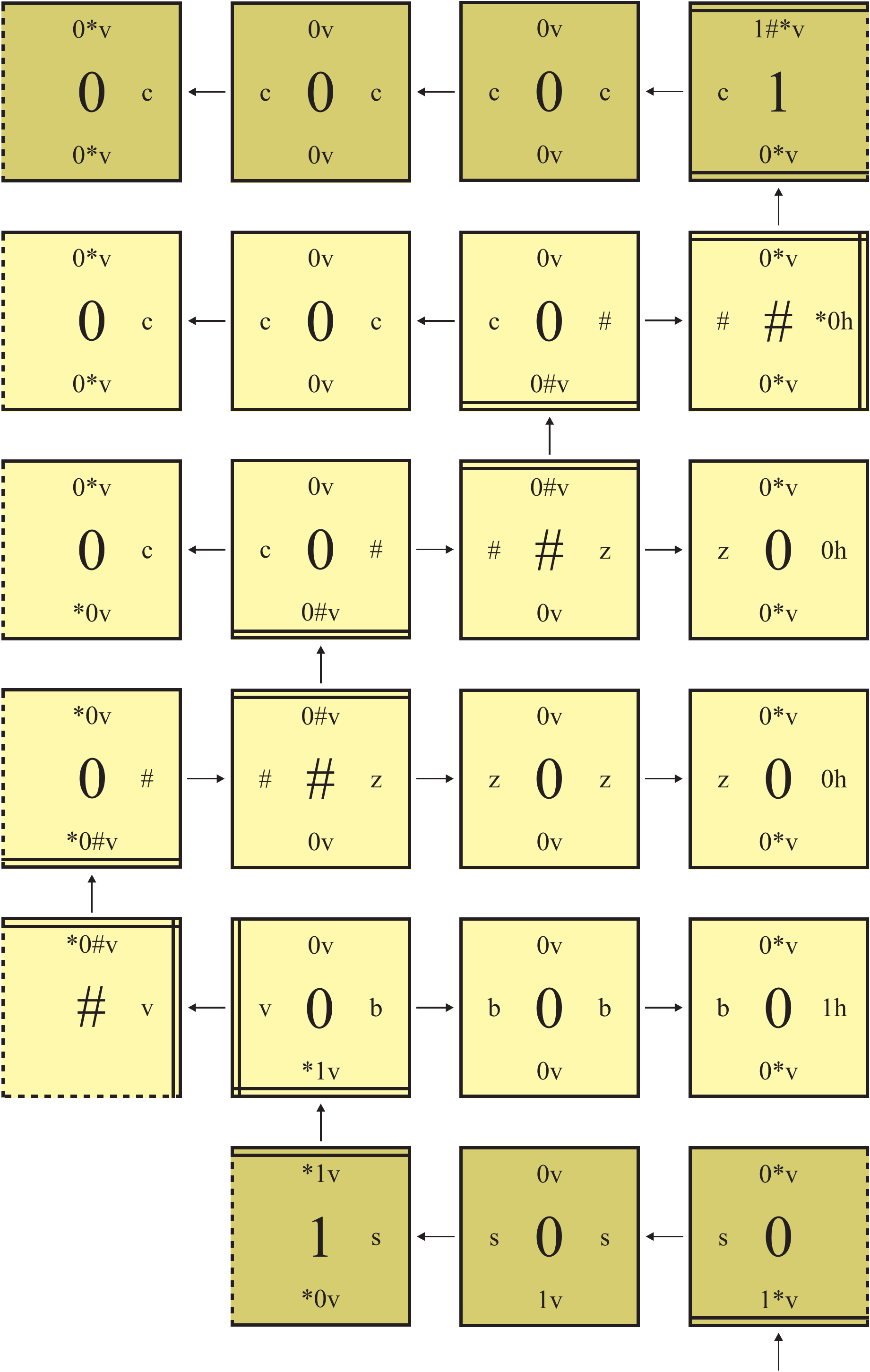}
\label{square_yaxis} \caption{The self-assembly of $S_2^{\uparrow}$
above the topmost row of $\widetilde{Y}_1$.}
\end{center}
\end{figure}

We now have the machinery to construct a directed TAS in which $Y$
strictly self-assembles.

\begin{lemma}
\label{techlemma3} There is a $\tau$-$T$-assembly sequence
$\vec{\alpha} = (\alpha_i | \; 0 \leq i < k)$, with $\alpha =
\res{\vec{\alpha}}$, satisfying
\begin{enumerate}
\item $\alpha_0 = \sigma$, where, for all $\vec{m} \in \mathbb{Z}^2$,
$$
\sigma(\vec{m}) = \left\{
\begin{array}{ll}
\textmd{the tile type shown in Figure~15} & \textmd{if } \vec{m} = (0,1) \\
\uparrow & \textmd{otherwise,}
\end{array}
\right.
$$
\item $\dom{\alpha} = Y$,
\item $\alpha$ is locally deterministic, and
\item for all $\vec{m} \in \dom{\alpha}$,
$$
\alpha\left(\vec{m}\right)(1,0) = \left\{
\begin{array}{ll}
\left(\textmd{*1h},2\right) & \textmd{if } \exists j \in \mathbb{N}, \textmd{ and }  \vec{m} = (0,\theta(2j+1))  \\
\left(\textmd{*0h},2\right) & \textmd{if } \exists j \in \mathbb{N}, \textmd{ and } \vec{m} = (0,\theta(2j)) \\
\left(\textmd{0h},1\right) & \textmd{if } \exists j \in \mathbb{N}, \textmd{ and } \\
                           & \quad\vec{m} \in \left(\{0\} \times \{\theta\left(j\right)-\rho\left(j\right)+1, \ldots, \theta\left(j\right)-1\}\right)  \\
\left(\textmd{1h},1\right) & \textmd{if } \exists j \in \mathbb{N}, \textmd{ and }  \vec{m} = (0,\theta(j)-\rho(j))  \\
\left(\lambda,0\right) & \textmd{otherwise}. \\
\end{array}\right.
$$
\end{enumerate}
\end{lemma}
\begin{proof}
Simply combine Lemmas~\ref{techlemma1}, and~\ref{techlemma2} to get
a locally deterministic assembly sequence for $Y$.
\end{proof}
\begin{figure}[htp]
\begin{center}
\includegraphics[width=0.54in]{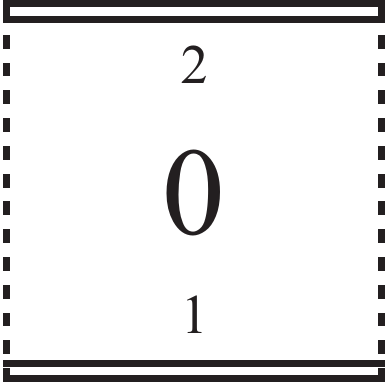}
\label{yseedtile} \caption{The seed tile type for the $y$-axis of
$\mathbf{T}$.}
\end{center}
\end{figure}

\begin{theorem}
$Y$ strictly self-assembles in the directed TAS $\mathcal{T} =
(T,\sigma,\tau)$.
\end{theorem}
\begin{proof}
Lemma~\ref{techlemma3} testifies to the fact that $\mathcal{T}$ is a
locally deterministic TAS, and hence is directed.
\end{proof}

A straightforward ``reflection'' of $T$ will yield a directed TAS in
which $X$ strictly self-assembles.

\begin{corollary}
Let $h : \Sigma \rightarrow \Sigma$, where for all $a \in
\Sigma$,
$$
h(a) = \left\{
\begin{array}{ll}
\lambda & \textmd{if } a = \textmd{v} \\
a & \textmd{otherwise.}
\end{array}
\right.
$$
$X$ strictly self-assembles in the directed TAS $\mathcal{T}' =
(T',\sigma',\tau)$, where
$$
T^{\prime} = \left\{ t' \; | \textmd{ for all } (x,y)\in U_2, t \in
T, \; t'(y,x) = \left(
h(\textmd{rev}\left(\color_t(x,y)\right))\textmd{h},\strength_t(x,y)\right)
\right\},
$$
and, for all $\vec{m} \in \mathbb{Z} \times \mathbb{Z}$ and $(x,y)
\in U_2$,
$$
\sigma'(\vec{m})(y,x) = \left\{
\begin{array}{ll}
\sigma(\vec{m})(x,y) & \textmd{if } \vec{m} = (1,0) \\
\uparrow & \textmd{otherwise.}
\end{array}
\right.
$$
\end{corollary}

\subsection{Self-Assembly of the Interior}
We now turn our attention to the self-assembly of the interior of
$\mathbf{T}$.

In the following lemma, we show how vertical bars attach to the
$X$-axis.

\begin{lemma}
\label{techlemma4} Let $j \in \mathbb{N}$. If $\sigma \in
\asmbt{\tau}{T \cup T'}$ satisfies
\begin{enumerate}
\item $\displaystyle\dom{\sigma} = X$, and
\item for all $\vec{m} \in \left(\{ \theta(j)-\rho(j),\ldots,\theta(j)\} \times \{0\}\right)$,
$$
\sigma\left(\vec{m}\right)(0,1) = \left\{
\begin{array}{ll}
\left(\textmd{1*v},2\right) & \textmd{if } j \textmd{ is odd} \\
\left(\textmd{0*v},2\right) & \textmd{if } j \textmd{ is even, and } \vec{m} = \left( \theta(j),0 \right)  \\
\left(\textmd{1v},1\right) & \textmd{if } j \textmd{ is even, and } \vec{m} = \left( \theta(j)-\rho(j),0 \right) \\
\left(\textmd{0v},1\right)  & \textmd{otherwise},
\end{array}\right.
$$
\end{enumerate}
then there is a $\tau$-$T \cup T'$-assembly sequence $\vec{\alpha} =
\left( \alpha_i \left| \; 0\leq i < k < \infty \right. \right)$,
with $\alpha = \res{\vec{\alpha}}$, satisfying
\begin{enumerate}
\item $\alpha_0 = \sigma$,
\item $\dom{\alpha} = X \cup \left( \theta_j(X) + Y_{\rho(j)-1}\right)$,
\item for all $\vec{m} \in \dom{\alpha} - \dom{\sigma}$,
$$
\sum_{\vec{u} \in
\textmd{IN}^{\vec{\alpha}}(\vec{m})}{\textmd{str}_{\alpha_{i_{\vec{\alpha}}(\vec{m})}}(\vec{m},\vec{u})
} = \tau,
$$
\item for all $\vec{m} \in \dom{\alpha} - \dom{\sigma}$ and
all $t \in \left(T \cup T'\right)-\{\alpha(\vec{m})\}$, $\vec{m}
\not \in \frontiertt{\left(\vec{\alpha} \setminus \vec{m}\right)}$,
and
\item for all $\vec{m} \in \dom{\alpha}$,
$$
\alpha\left(\vec{m}\right)(1,0) = \left\{
\begin{array}{ll}
\left(\textmd{*1h},2\right) & \textmd{if } \exists j' \in \mathbb{N}, \textmd{ and }  \vec{m} = (\theta(j),\theta(2j'+1))  \\
\left(\textmd{*0h},2\right) & \textmd{if } \exists j' \in \mathbb{N}, \textmd{ and } \vec{m} = (\theta(j),\theta(2j')) \\
\left(\textmd{0h},1\right) & \textmd{if } \exists j' \in \mathbb{N}, \textmd{ and } \\
                           & \vec{m} \in (\{\theta(j)\} \times \\
                         & \quad \{\theta\left(j'\right)-\rho\left(j'\right)+1, \ldots, \theta\left(j'\right)-1\})  \\
\left(\textmd{1h},1\right) & \textmd{if } \exists j' \in \mathbb{N} \textmd{ and }  \vec{m} = (\theta(j),\theta(j')-\rho(j'))  \\
\left(\lambda,0\right) & \textmd{otherwise}. \\
\end{array}\right.
$$
\end{enumerate}
\end{lemma}
\begin{proof}
This follows directly from Lemma~\ref{techlemma0}.
\end{proof}

\begin{corollary}
\label{techcorollary1} Let $j,j' \in \mathbb{N}$ with $1 \leq j' <
2^{\rho(j)}$. If $\sigma \in \asmbt{\tau}{T \cup T'}$ satisfies
\begin{enumerate}
\item $X \cup \left(\theta_j(X)+ Y_{\rho(j)-1}\right) \subseteq
\dom{\sigma} \subseteq \mathbf{T} - \left((\theta(j),\theta(j')) +
X_{\rho\left(j'\right)-1} \right)$,
\item for all $ \vec{m} \in
\left(\left(\{0\} \times
\{\theta\left(j'\right)-\rho\left(j'\right),\ldots,\theta\left(j'\right)\}
\right) + (\theta(j),0)\right)$,
$$
\alpha\left(\vec{m}\right)(1,0) = \left\{
\begin{array}{ll}
\left(\textmd{*1h},2\right) & \textmd{if } j' \textmd{ is odd}  \\
\left(\textmd{*0h},2\right) & \textmd{if } j' \textmd{ is even, and } \vec{m} = (\theta(j),\theta(j')) \\
\left(\textmd{1h},1\right) & \textmd{if } j' \textmd{ is even, and }  \vec{m} = (\theta(j),\theta(j')-\rho(j')) \\
\left(\textmd{0h},1\right) & \textmd{otherwise,} \\
\end{array}\right.
$$
\end{enumerate}
there is a $\tau$-$T \cup T'$-assembly sequence $\vec{\alpha} =
(\alpha_i \; | \; 0 \leq i < k < \infty)$ satisfying
\begin{enumerate}
\item $\alpha_0 = \sigma$,
\item $X \cup \left(\theta_j(X)+ Y_{\rho(j)-1}\right) \cup \left((\theta(j),\theta(j')) + X_{\rho\left(j'\right)-1} \right) \subseteq
\dom{\alpha} \subseteq \dom{\sigma} \cup
\left((\theta(j),\theta(j')) + X_{\rho\left(j'\right)-1} \right)$,
\item for all $\vec{m} \in \dom{\alpha} - \dom{\sigma}$,
$$
\sum_{\vec{u} \in
\textmd{IN}^{\vec{\alpha}}(\vec{m})}{\textmd{str}_{\alpha_{i_{\vec{\alpha}}(\vec{m})}}(\vec{m},\vec{u})
} = \tau,
$$
\item for all $\vec{m} \in \dom{\alpha} - \dom{\sigma}$ and
all $t \in \left(T \cup T'\right)-\{\alpha(\vec{m})\}$, $\vec{m}
\not \in \frontiertt{\left(\vec{\alpha} \setminus \vec{m}\right)}$,
and
\item for all $\vec{m} \in \dom{\alpha} - \dom{\sigma}$,
$$
\alpha\left(\vec{m}\right)(1,0) = \left\{
\begin{array}{ll}
\left(\textmd{*1v},2\right) & \textmd{if } \exists j'' \in \mathbb{N}, \textmd{ and }  \vec{m} = (\theta(2j''+1),\theta(j'))  \\
\left(\textmd{*0v},2\right) & \textmd{if } \exists j'' \in \mathbb{N}, \textmd{ and } \vec{m} = (\theta(2j''),\theta(j')) \\
\left(\textmd{0v},1\right) & \textmd{if } \exists j'' \in \mathbb{N}, \textmd{ and } \\
                            & \vec{m} \in (\{\theta\left(j''\right)-\rho\left(j''\right)+1, \ldots, \theta\left(j''\right)-1\} \\
                           & \quad \times \{\theta(j')\})  \\
\left(\textmd{1v},1\right) & \textmd{if } \exists j'' \in \mathbb{N}, \textmd{ and }  \vec{m} = (\theta(j'')-\rho(j''),\theta(j'))  \\
\left(\lambda,0\right) & \textmd{otherwise}. \\
\end{array}\right.
$$
\end{enumerate}
\end{corollary}

Note that the results of this subsection are invariant under
``reflection.''

\subsection{Proof of Correctness}
We are now ready to prove our second main theorem.

\begin{lemma}
\label{techlemma5} Let
$$
T_{\mathbf{T}} = T \cup T' \cup \{ \textmd{the tile type shown in
Figure~16} \}.
$$
There is a $\tau$-$T_{\mathbf{T}}$-assembly sequence $\vec{\alpha} =
(\alpha_i \; | \; 0 \leq i < k)$, with $\alpha =
\res{\vec{\alpha}}$, satisfying
\begin{enumerate}
\item $\alpha_0 = \sigma_{\mathbf{T}}$, where, for all $\vec{m} \in \mathbb{Z}
\times \mathbb{Z}$,
$$
\sigma_{\mathbf{T}}(\vec{m}) = \left\{
\begin{array}{ll}
\textmd{the tile type shown in Figure~16} & \textmd{if } \vec{m} = (0,0) \\
\uparrow & \textmd{otherwise,}
\end{array}
\right.
$$
\item $\dom{\alpha} = \mathbf{T}$, and
\item $\vec{\alpha}$ is locally deterministic.
\end{enumerate}
\end{lemma}
\begin{proof}
Simply dovetail the assembly sequences given by
Lemmas~\ref{techlemma3},~\ref{techlemma4}, and
Corollary~\ref{techcorollary1}, to get a locally deterministic
assembly sequence for $\mathbf{T}$.
\end{proof}

\begin{figure}[htp]
\begin{center}
\includegraphics[width=0.54in]{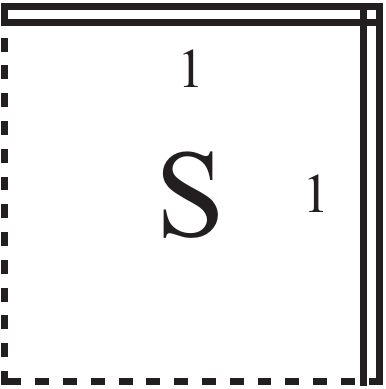}
\label{seedtile} \caption{The single seed tile type for
$\mathbf{T}$.}
\end{center}
\end{figure}

\begin{theorem}
\label{secondmaintheorem} $\mathbf{T}$ strictly self-assembles in
the directed TAS $\mathcal{T}_{\mathbf{T}} =
(T_{\mathbf{T}},\sigma_{\mathbf{T}},\tau)$.
\end{theorem}
\begin{proof}
This follows immediately from Lemma~\ref{techlemma5}.
\end{proof}

\subsubsection*{Acknowledgment}
We thank Dave Doty, Xiaoyang Gu, Satya Nandakumar, John Mayfield,
Matt Patitz, Aaron Sterling, and Kirk Sykora for useful discussions.

%----------------------------------------------------------------
\bibliographystyle{amsplain}
\bibliography{main,dim,random,dimrelated,rbm}

\end{document}